\documentclass{article}

\usepackage[english]{babel}
\usepackage{amsmath,amsthm,amssymb,mathabx}
\usepackage[svgnames]{xcolor}
\usepackage{rotating}
\usepackage{floatpag}
\usepackage{xifthen}
\usepackage{bm}
\usepackage{cite}
\usepackage{booktabs}

\usepackage{xspace,enumerate,color,epsfig}
\usepackage{graphicx}
\graphicspath{{.}{./figures/}}

\usepackage{tikzfig}
\usepackage{stmaryrd}
\usepackage{docmute}
\usepackage{keycommand}

\newkeycommand{\boxmap}[style=small box][1]{\,\tikz{\node[style=\commandkey{style}] (x) {$#1$};\draw(0,-1.25)--(x)--(0,1.25);}\,}
\newkeycommand{\dboxmap}[style=dbox][1]{\,\tikz{\node[style=\commandkey{style}] (x) {$#1$};\draw[boldedge] (0,-1.25)--(x)--(0,1.25);}\,}

\newkeycommand{\wirelabel}[style=white label]{\,\tikz{\node[style=\commandkey{style}]{};}\,\xspace}

\newkeycommand{\pointmap}[style=point][1]{\,\tikz{\node[style=\commandkey{style}] (x) at (0,-0.3) {$#1$};\node [style=none] (3) at (0, 0.9) {};\node [style=none] (3) at (0, -0.9) {};\draw (x)--(0,0.7);}\,}

\newkeycommand{\point}[style=point,estyle=][1]{\,\tikz{\node[style=\commandkey{style}] (x) at (0,-0.05) {$#1$};\draw [\commandkey{estyle}] (x)--(0,1.05);}\,}

\newkeycommand{\copointmap}[style=copoint][1]{\,\tikz{\node[style=\commandkey{style}] (x) at (0,0.3) {$#1$};\draw (x)--(0,-0.7);}\,}

\newkeycommand{\dpoint}[style=point][1]{\,\tikz{\node[style=\commandkey{style},doubled] (x) at (0,-0.3) {$#1$};\draw[boldedge] (x)--(0,0.7);}\,}

\newkeycommand{\kpoint}[style=kpoint,estyle=][1]{\,\tikz{\node[style=\commandkey{style}] (x) at (0,-0.05) {$#1$};\draw [\commandkey{estyle}] (x)--(0,1.05);}\,}

\newkeycommand{\kpointgray}[style=gray kpoint,estyle=][1]{\,\tikz{\node[style=\commandkey{style}] (x) at (0,-0.05) {$#1$};\draw [\commandkey{estyle}] (x)--(0,1.05);}\,}

\newkeycommand{\typedkpoint}[style=kpoint,edgestyle=][2]{%
\,\begin{tikzpicture}
  \begin{pgfonlayer}{nodelayer}
    \node [style=none] (0) at (0, 1) {};
    \node [style=\commandkey{style}] (1) at (0, -0.75) {$#2$};
    \node [style=right label] (2) at (0.25, 0.5) {$#1$};
  \end{pgfonlayer}
  \begin{pgfonlayer}{edgelayer}
    \draw [\commandkey{edgestyle}] (1) to (0.center);
  \end{pgfonlayer}
\end{tikzpicture}\,}

\newkeycommand{\typedkpointdag}[style=kpoint adjoint,edgestyle=][2]{%
\,\begin{tikzpicture}
  \begin{pgfonlayer}{nodelayer}
    \node [style=none] (0) at (0, -1) {};
    \node [style=\commandkey{style}] (1) at (0, 0.75) {$#2$};
    \node [style=right label] (2) at (0.25, -0.5) {$#1$};
  \end{pgfonlayer}
  \begin{pgfonlayer}{edgelayer}
    \draw [\commandkey{edgestyle}] (0.center) to (1);
  \end{pgfonlayer}
\end{tikzpicture}\,}

\newkeycommand{\bistate}[style=kpoint][1]{\,\begin{tikzpicture}
  \begin{pgfonlayer}{nodelayer}
    \node [style=\commandkey{style}, minimum width=1 cm, inner sep=2pt] (0) at (0, -0.25) {$#1$};
    \node [style=none] (1) at (-0.75, 0) {};
    \node [style=none] (2) at (0.75, 0) {};
    \node [style=none] (3) at (-0.75, 0.88) {};
    \node [style=none] (4) at (0.75, 0.88) {};
  \end{pgfonlayer}
  \begin{pgfonlayer}{edgelayer}
    \draw (1.center) to (3.center);
    \draw (2.center) to (4.center);
  \end{pgfonlayer}
\end{tikzpicture}\,}

\newkeycommand{\bistatebig}[style=kpoint][1]{\,\begin{tikzpicture}
  \begin{pgfonlayer}{nodelayer}
    \node [style=\commandkey{style}, minimum width=, minimum width=1.26 cm, inner sep=2pt] (0) at (0, -0.25) {$#1$};
    \node [style=none] (1) at (-1, 0) {};
    \node [style=none] (2) at (1, 0) {};
    \node [style=none] (3) at (-1, 0.88) {};
    \node [style=none] (4) at (1, 0.88) {};
  \end{pgfonlayer}
  \begin{pgfonlayer}{edgelayer}
    \draw (1.center) to (3.center);
    \draw (2.center) to (4.center);
  \end{pgfonlayer}
\end{tikzpicture}\,}

\newkeycommand{\dbistate}[style=dkpoint][1]{\,\begin{tikzpicture}
  \begin{pgfonlayer}{nodelayer}
    \node [style=\commandkey{style}, minimum width=1 cm, inner sep=2pt] (0) at (0, -0.25) {$#1$};
    \node [style=none] (1) at (-0.75, 0) {};
    \node [style=none] (2) at (0.75, 0) {};
    \node [style=none] (3) at (-0.75, 1.0) {};
    \node [style=none] (4) at (0.75, 1.0) {};
  \end{pgfonlayer}
  \begin{pgfonlayer}{edgelayer}
    \draw [boldedge] (1.center) to (3.center);
    \draw [boldedge] (2.center) to (4.center);
  \end{pgfonlayer}
\end{tikzpicture}\,}

\newkeycommand{\bistatebraket}[style1=kpoint,style2=kpointadj][2]{\,\begin{tikzpicture}
  \begin{pgfonlayer}{nodelayer}
    \node [style=\commandkey{style1}, minimum width=1 cm, inner sep=2pt] (0) at (0, -0.75) {$#1$};
    \node [style=none] (1) at (-0.75, -0.5) {};
    \node [style=none] (2) at (0.75, -0.5) {};
    \node [style=none] (3) at (-0.75, 0.5) {};
    \node [style=none] (4) at (0.75, 0.5) {};
    \node [style=\commandkey{style2}, minimum width=1 cm, inner sep=2pt] (5) at (0, 0.75) {$#2$};
  \end{pgfonlayer}
  \begin{pgfonlayer}{edgelayer}
    \draw (1.center) to (3.center);
    \draw (2.center) to (4.center);
  \end{pgfonlayer}
\end{tikzpicture}\,}

\newkeycommand{\weight}[style=dkpoint][1]{\,\begin{tikzpicture}
  \begin{pgfonlayer}{nodelayer}
    \node [style=\commandkey{style}] (0) at (0, -0.5) {$#1$};
    \node [style=upground] (1) at (0, 0.75) {};
  \end{pgfonlayer}
  \begin{pgfonlayer}{edgelayer}
    \draw [style=boldedge] (0) to (1);
  \end{pgfonlayer}
\end{tikzpicture}\,}

\newkeycommand{\dkpoint}[style=dkpoint][1]{\,\tikz{\node[style=\commandkey{style}] (x) at (0,-0.1) {$#1$};\draw[boldedge] (x)--(0,1);}\,}
\newkeycommand{\dkpointadj}[style=dkpointadj][1]{\,\tikz{\node[style=\commandkey{style}] (x) at (0,0.1) {$#1$};\draw[boldedge] (x)--(0,-1);}\,}
\newkeycommand{\dkpointtrans}[style=dkpointtrans][1]{\,\tikz{\node[style=\commandkey{style}] (x) at (0,0.1) {$#1$};\draw[boldedge] (x)--(0,-1);}\,}
\newkeycommand{\dkpointconj}[style=dkpointconj][1]{\,\tikz{\node[style=\commandkey{style}] (x) at (0,-0.1) {$#1$};\draw[boldedge] (x)--(0,1);}\,}

\newkeycommand{\blackkpoint}[style=black kpoint][1]{\,\tikz{\node[style=\commandkey{style}] (x) at (0,-0.1) {$#1$};\draw (x)--(0,1);}\,}
\newkeycommand{\blackkpointadj}[style=black kpointadj][1]{\,\tikz{\node[style=\commandkey{style}] (x) at (0,0.1) {$#1$};\draw (x)--(0,-1);}\,}
\newkeycommand{\blackdkpoint}[style=black dkpoint][1]{\,\tikz{\node[style=\commandkey{style}] (x) at (0,-0.1) {$#1$};\draw[boldedge] (x)--(0,1);}\,}
\newkeycommand{\blackdkpointadj}[style=black dkpointadj][1]{\,\tikz{\node[style=\commandkey{style}] (x) at (0,0.1) {$#1$};\draw[boldedge] (x)--(0,-1);}\,}

\newcommand{\trace}{\,\begin{tikzpicture}
  \begin{pgfonlayer}{nodelayer}
    \node [style=none] (0) at (0, -0.60) {};
    \node [style=upground] (1) at (0, 0.40) {};
  \end{pgfonlayer}
  \begin{pgfonlayer}{edgelayer}
    \draw [style=boldedge] (0.center) to (1);
  \end{pgfonlayer}
\end{tikzpicture}\,}

\newcommand\discard\trace

\newkeycommand{\pointketbra}[style1=copoint,style2=point,estyle=][2]{\,%
\begin{tikzpicture}
  \begin{pgfonlayer}{nodelayer}
    \node [style=none] (0) at (0, -1.75) {};
    \node [style=\commandkey{style1}] (1) at (0, -0.75) {$#1$};
    \node [style=\commandkey{style2}] (2) at (0, 0.75) {$#2$};
    \node [style=none] (3) at (0, 1.75) {};
  \end{pgfonlayer}
  \begin{pgfonlayer}{edgelayer}
    \draw [\commandkey{estyle}] (0.center) to (1);
    \draw [\commandkey{estyle}] (2) to (3.center);
  \end{pgfonlayer}
\end{tikzpicture}\,}

\newkeycommand{\twocopointketbra}[style1=copoint,style2=copoint,style3=point][3]{\,%
\begin{tikzpicture}
  \begin{pgfonlayer}{nodelayer}
    \node [style=none] (0) at (-0.75, -1.75) {};
    \node [style=none] (1) at (0.75, -1.75) {};
    \node [style=\commandkey{style1}] (2) at (-0.75, -0.75) {$#1$};
    \node [style=\commandkey{style2}] (3) at (0.75, -0.75) {$#2$};
    \node [style=\commandkey{style3}] (4) at (0, 0.75) {$#3$};
    \node [style=none] (5) at (0, 1.75) {};
  \end{pgfonlayer}
  \begin{pgfonlayer}{edgelayer}
    \draw (0.center) to (2);
    \draw (1.center) to (3);
    \draw (4) to (5.center);
  \end{pgfonlayer}
\end{tikzpicture}\,}

\newkeycommand{\twopointketbra}[style1=copoint,style2=point,style3=point][3]{\,%
\begin{tikzpicture}
  \begin{pgfonlayer}{nodelayer}
    \node [style=none] (0) at (0, -1.75) {};
    \node [style=none] (1) at (-0.75, 1.75) {};
    \node [style=\commandkey{style1}] (2) at (0, -0.75) {$#1$};
    \node [style=\commandkey{style2}] (3) at (-0.75, 0.75) {$#2$};
    \node [style=\commandkey{style3}] (4) at (0.75, 0.75) {$#3$};
    \node [style=none] (5) at (0.75, 1.75) {};
  \end{pgfonlayer}
  \begin{pgfonlayer}{edgelayer}
    \draw (0.center) to (2);
    \draw (1.center) to (3);
    \draw (4) to (5.center);
  \end{pgfonlayer}
\end{tikzpicture}\,}

\newkeycommand{\pointbraket}[style1=point,style2=copoint,estyle=][2]{\,%
\begin{tikzpicture}
  \begin{pgfonlayer}{nodelayer}
    \node [style=\commandkey{style1}] (0) at (0, -0.625) {$#1$};
    \node [style=\commandkey{style2}] (1) at (0, 0.625) {$#2$};
  \end{pgfonlayer}
  \begin{pgfonlayer}{edgelayer}
    \draw [\commandkey{estyle}] (0) to (1);
  \end{pgfonlayer}
\end{tikzpicture}\,}

\newkeycommand{\kpointketbra}[style1=kpointdag,style2=kpoint,estyle=][2]{\,%
\begin{tikzpicture}
  \begin{pgfonlayer}{nodelayer}
    \node [style=none] (0) at (0, -1.9) {};
    \node [style=\commandkey{style1}] (1) at (0, -0.95) {$#1$};
    \node [style=\commandkey{style2}] (2) at (0, 0.95) {$#2$};
    \node [style=none] (3) at (0, 1.9) {};
  \end{pgfonlayer}
  \begin{pgfonlayer}{edgelayer}
    \draw [\commandkey{estyle}] (0.center) to (1);
    \draw [\commandkey{estyle}] (2) to (3.center);
  \end{pgfonlayer}
\end{tikzpicture}\,}

\newkeycommand{\kpointmap}[style1=kpoint,style2=map][2]{\,%
\begin{tikzpicture}
  \begin{pgfonlayer}{nodelayer}
    \node [style=\commandkey{style1}] (0) at (0, -1.1) {$#1$};
    \node [style=\commandkey{style2}] (1) at (0, 0.5) {$#2$};
    \node [style=none] (2) at (0, 1.75) {};
  \end{pgfonlayer}
  \begin{pgfonlayer}{edgelayer}
    \draw (0) to (1);
    \draw (1) to (2.center);
  \end{pgfonlayer}
\end{tikzpicture}%
\,}

\newkeycommand{\kpointmaptrans}[style1=kpointconj,style2=mapadj][2]{\,%
\begin{tikzpicture}
  \begin{pgfonlayer}{nodelayer}
    \node [style=\commandkey{style1}] (0) at (0, -1.1) {$#1$};
    \node [style=\commandkey{style2}] (1) at (0, 0.5) {$#2$};
    \node [style=none] (2) at (0, 1.75) {};
  \end{pgfonlayer}
  \begin{pgfonlayer}{edgelayer}
    \draw (0) to (1);
    \draw (1) to (2.center);
  \end{pgfonlayer}
\end{tikzpicture}%
\,}

\newkeycommand{\kpointmapadj}[style1=kpointadj,style2=map][2]{\,%
\begin{tikzpicture}
  \begin{pgfonlayer}{nodelayer}
    \node [style=\commandkey{style1}] (0) at (0, 1.1) {$#1$};
    \node [style=\commandkey{style2}] (1) at (0, -0.5) {$#2$};
    \node [style=none] (2) at (0, -1.75) {};
  \end{pgfonlayer}
  \begin{pgfonlayer}{edgelayer}
    \draw (0) to (1);
    \draw (1) to (2.center);
  \end{pgfonlayer}
\end{tikzpicture}%
\,}

\newkeycommand{\dkpointmap}[style1=dkpoint,style2=dmap][2]{\,%
\begin{tikzpicture}
  \begin{pgfonlayer}{nodelayer}
    \node [style=\commandkey{style1}] (0) at (0, -1.2) {$#1$};
    \node [style=\commandkey{style2}] (1) at (0, 0.4) {$#2$};
    \node [style=none] (2) at (0, 1.7) {};
  \end{pgfonlayer}
  \begin{pgfonlayer}{edgelayer}
    \draw[style=boldedge] (0) to (1);
    \draw[style=boldedge] (1) to (2.center);
  \end{pgfonlayer}
\end{tikzpicture}%
\,}

\newkeycommand{\dkpointmapx}[style1=dkpoint,style2=dmap][2]{\,%
\begin{tikzpicture}
  \begin{pgfonlayer}{nodelayer}
    \node [style=\commandkey{style1}] (0) at (0, -1.4) {$#1$};
    \node [style=\commandkey{style2}] (1) at (0, 0.6) {$#2$};
    \node [style=none] (2) at (0, 1.9) {};
  \end{pgfonlayer}
  \begin{pgfonlayer}{edgelayer}
    \draw[style=boldedge] (0) to (1);
    \draw[style=boldedge] (1) to (2.center);
  \end{pgfonlayer}
\end{tikzpicture}%
\,}

\newkeycommand{\boxcopointmapdag}[style1=map,style2=copoint][2]{\,%
\begin{tikzpicture}
  \begin{pgfonlayer}{nodelayer}
    \node [style=none] (0) at (0, -1.75) {};
    \node [style=\commandkey{style1}] (1) at (0, -0.5) {$#1$};
    \node [style=\commandkey{style2}] (2) at (0, 1) {$#2$};
  \end{pgfonlayer}
  \begin{pgfonlayer}{edgelayer}
    \draw (0.center) to (1);
    \draw (1) to (2);
  \end{pgfonlayer}
\end{tikzpicture}%
\,}

\newkeycommand{\kpointdagmap}[style1=map,style2=kpointdag][2]{\,%
\begin{tikzpicture}
  \begin{pgfonlayer}{nodelayer}
    \node [style=none] (0) at (0, -1.75) {};
    \node [style=\commandkey{style1}] (1) at (0, -0.5) {$#1$};
    \node [style=\commandkey{style2}] (2) at (0, 1.1) {$#2$};
  \end{pgfonlayer}
  \begin{pgfonlayer}{edgelayer}
    \draw (0.center) to (1);
    \draw (1) to (2);
  \end{pgfonlayer}
\end{tikzpicture}%
\,}

\newkeycommand{\sandwichmap}[style1=point,style2=map,style3=copoint][3]{\,%
\begin{tikzpicture}
  \begin{pgfonlayer}{nodelayer}
    \node [style=\commandkey{style1}] (0) at (0, -1.50) {$#1$};
    \node [style=\commandkey{style2}] (1) at (0, 0) {$#2$};
    \node [style=\commandkey{style3}] (2) at (0, 1.50) {$#3$};
  \end{pgfonlayer}
  \begin{pgfonlayer}{edgelayer}
    \draw (0) to (1);
    \draw (1) to (2);
  \end{pgfonlayer}
\end{tikzpicture}%
}

\newkeycommand{\kpointsandwichmap}[style1=kpoint,style2=map,style3=kpointdag][3]{\,%
\begin{tikzpicture}
  \begin{pgfonlayer}{nodelayer}
    \node [style=\commandkey{style1}] (0) at (0, -1.5) {$#1$};
    \node [style=\commandkey{style2}] (1) at (0, 0) {$#2$};
    \node [style=\commandkey{style3}] (2) at (0, 1.5) {$#3$};
  \end{pgfonlayer}
  \begin{pgfonlayer}{edgelayer}
    \draw (0) to (1);
    \draw (1) to (2);
  \end{pgfonlayer}
\end{tikzpicture}%
}

\newkeycommand{\sandwichmapconj}[style1=point,style2=mapconj,style3=copoint][3]{\,%
\begin{tikzpicture}
  \begin{pgfonlayer}{nodelayer}
    \node [style=\commandkey{style1}] (0) at (0, -1.50) {$#1$};
    \node [style=\commandkey{style2}] (1) at (0, 0) {$#2$};
    \node [style=\commandkey{style3}] (2) at (0, 1.50) {$#3$};
  \end{pgfonlayer}
  \begin{pgfonlayer}{edgelayer}
    \draw (0) to (1);
    \draw (1) to (2);
  \end{pgfonlayer}
\end{tikzpicture}%
}

\newkeycommand{\sandwichtwo}[style1=point,style2=map,style3=map,style4=copoint][4]{\,%
\begin{tikzpicture}
  \begin{pgfonlayer}{nodelayer}
    \node [style=\commandkey{style2}] (0) at (0, -1) {$#2$};
    \node [style=\commandkey{style3}] (1) at (0, 1) {$#3$};
    \node [style=\commandkey{style1}] (2) at (0, -2.6) {$#1$};
    \node [style=\commandkey{style4}] (3) at (0, 2.6) {$#4$};
  \end{pgfonlayer}
  \begin{pgfonlayer}{edgelayer}
    \draw (2) to (0);
    \draw (0) to (1);
    \draw (1) to (3);
  \end{pgfonlayer}
\end{tikzpicture}\,}

\newkeycommand{\longbraket}[style1=point,style2=copoint][2]{\,%
\begin{tikzpicture}
  \begin{pgfonlayer}{nodelayer}
    \node [style=\commandkey{style1}] (0) at (0, -2.6) {$#1$};
    \node [style=\commandkey{style2}] (1) at (0, 2.6) {$#2$};
  \end{pgfonlayer}
  \begin{pgfonlayer}{edgelayer}
    \draw (0) to (1);
  \end{pgfonlayer}
\end{tikzpicture}\,}

\newkeycommand{\onb}[style=point]{\ensuremath{\{\,\tikz{\node[style=\commandkey{style}] (x) at (0,-0.3) {$j$};\draw (x)--(0,0.7);}\,\}}\xspace}

\newkeycommand{\phase}[style=white phase dot][1]{\,\begin{tikzpicture}
    \begin{pgfonlayer}{nodelayer}
        \node [style=none] (0) at (0, 1) {};
        \node [style=\commandkey{style}] (2) at (0, -0) {$#1$}; 
        \node [style=none] (3) at (0, -1) {};
    \end{pgfonlayer}
    \begin{pgfonlayer}{edgelayer}
        \draw (2) to (0.center);
        \draw (3.center) to (2);
    \end{pgfonlayer}
\end{tikzpicture}\,}

\newkeycommand{\dphase}[style=white phase ddot][1]{\,\begin{tikzpicture}
    \begin{pgfonlayer}{nodelayer}
        \node [style=none] (0) at (0, 1.25) {};
        \node [style=\commandkey{style}] (2) at (0, -0) {$#1$};
        \node [style=none] (3) at (0, -1.25) {};
    \end{pgfonlayer}
    \begin{pgfonlayer}{edgelayer}
        \draw [boldedge] (2) to (0.center);
        \draw [boldedge] (3.center) to (2);
    \end{pgfonlayer}
\end{tikzpicture}\,}

\newkeycommand{\dphasepoint}[style=white phase ddot][1]{\,\begin{tikzpicture}[yshift=-3mm]
  \begin{pgfonlayer}{nodelayer}
    \node [style=none] (0) at (0, 1) {};
    \node [style=\commandkey{style}] (1) at (0, 0) {$#1$};
  \end{pgfonlayer}
  \begin{pgfonlayer}{edgelayer}
    \draw [style=boldedge] (0.center) to (1);
  \end{pgfonlayer}
\end{tikzpicture}\,}

\newkeycommand{\dphasepointgray}[style=gray phase ddot][1]{\,\begin{tikzpicture}[yshift=-3mm]
  \begin{pgfonlayer}{nodelayer}
    \node [style=none] (0) at (0, 1) {};
    \node [style=\commandkey{style}] (1) at (0, 0) {$#1$};
  \end{pgfonlayer}
  \begin{pgfonlayer}{edgelayer}
    \draw [style=boldedge] (0.center) to (1);
  \end{pgfonlayer}
\end{tikzpicture}\,}

\newkeycommand{\dphasecopoint}[style=white phase ddot][1]{\,\begin{tikzpicture}[yshift=3mm,yscale=-1]
  \begin{pgfonlayer}{nodelayer}
    \node [style=none] (0) at (0, 1) {};
    \node [style=\commandkey{style}] (1) at (0, 0) {$#1$};
  \end{pgfonlayer}
  \begin{pgfonlayer}{edgelayer}
    \draw [style=boldedge] (0.center) to (1);
  \end{pgfonlayer}
\end{tikzpicture}\,}

\newkeycommand{\dphasepointsm}[style=white phase ddot][1]{\,\begin{tikzpicture}[yshift=-3mm]
  \begin{pgfonlayer}{nodelayer}
    \node [style=none] (0) at (0, 1) {};
    \node [style=\commandkey{style}] (1) at (0, 0) {\footnotesize\!$#1$\!};
  \end{pgfonlayer}
  \begin{pgfonlayer}{edgelayer}
    \draw [style=boldedge] (0.center) to (1);
  \end{pgfonlayer}
\end{tikzpicture}\,}

\newkeycommand{\phasepoint}[style=white phase dot][1]{\,\begin{tikzpicture}[yshift=-3mm]
  \begin{pgfonlayer}{nodelayer}
    \node [style=none] (0) at (0, 1) {};
    \node [style=\commandkey{style}] (1) at (0, 0) {$#1$};
  \end{pgfonlayer}
  \begin{pgfonlayer}{edgelayer}
    \draw (0.center) to (1);
  \end{pgfonlayer}
\end{tikzpicture}\,}

\newkeycommand{\phasecopoint}[style=white phase dot][1]{\,\begin{tikzpicture}[yshift=3mm]
  \begin{pgfonlayer}{nodelayer}
    \node [style=none] (0) at (0, -1) {};
    \node [style=\commandkey{style}] (1) at (0, 0) {$#1$};
  \end{pgfonlayer}
  \begin{pgfonlayer}{edgelayer}
    \draw (0.center) to (1);
  \end{pgfonlayer}
\end{tikzpicture}\,}

\newkeycommand{\kpointbraket}[style1=kpoint,style2=kpoint adjoint,estyle=][2]{\,%
\begin{tikzpicture}
  \begin{pgfonlayer}{nodelayer}
    \node [style=\commandkey{style1}] (0) at (0, -0.735) {$#1$};
    \node [style=\commandkey{style2}] (1) at (0, 0.735) {$#2$};
  \end{pgfonlayer}
  \begin{pgfonlayer}{edgelayer}
    \draw [\commandkey{estyle}] (0) to (1);
  \end{pgfonlayer}
\end{tikzpicture}\,}

\newkeycommand{\kkpointbraket}[style1=kpoint,style2=kpoint adjoint,estyle=][2]{\,%
\begin{tikzpicture}
  \begin{pgfonlayer}{nodelayer}
    \node [style=\commandkey{style1}] (0) at (0, -0.685) {$#1$};
    \node [style=\commandkey{style2}] (1) at (0, 0.685) {$#2$};
  \end{pgfonlayer}
  \begin{pgfonlayer}{edgelayer}
    \draw [\commandkey{estyle}] (0) to (1);
  \end{pgfonlayer}
\end{tikzpicture}\,}

\newkeycommand{\kpointbraketx}[style1=kpoint,style2=kpoint adjoint,estyle=][2]{\,%
\begin{tikzpicture}
  \begin{pgfonlayer}{nodelayer}
    \node [style=\commandkey{style1}] (0) at (0, -0.885) {$#1$};
    \node [style=\commandkey{style2}] (1) at (0, 0.885) {$#2$};
  \end{pgfonlayer}
  \begin{pgfonlayer}{edgelayer}
    \draw [\commandkey{estyle}] (0) to (1);
  \end{pgfonlayer}
\end{tikzpicture}\,}

\tikzstyle{cdiag}=[matrix of math nodes, row sep=3em, column sep=3em, text height=1.5ex, text depth=0.25ex,inner sep=0.5em]
\tikzstyle{arrow above}=[transform canvas={yshift=0.5ex}]
\tikzstyle{arrow below}=[transform canvas={yshift=-0.5ex}]

\usepackage[svgnames]{xcolor}
\usepackage{tikz}
\usetikzlibrary{decorations.markings}
\usetikzlibrary{shapes.geometric}
\pgfdeclarelayer{edgelayer}
\pgfdeclarelayer{nodelayer}
\pgfsetlayers{edgelayer,nodelayer,main}

\tikzstyle{none}=[inner sep=0pt]
\definecolor{hexcolor0xff0000}{rgb}{1.000,0.000,0.000}
\definecolor{hexcolor0x000000}{rgb}{0.000,0.000,0.000}
\definecolor{hexcolor0x00ff00}{rgb}{0.000,1.000,0.000}
\definecolor{hexcolor0x000000}{rgb}{0.000,0.000,0.000}
\definecolor{hexcolor0xffff00}{rgb}{1.000,1.000,0.000}
\definecolor{hexcolor0xffffff}{rgb}{1.000,1.000,1.000}

\tikzstyle{rn}=[circle,fill=hexcolor0xff0000,draw=hexcolor0x000000,line width=0.8 pt]
\tikzstyle{gn}=[circle,fill=hexcolor0x00ff00,draw=hexcolor0x000000,line width=0.8 pt]
\tikzstyle{yn}=[circle,fill=hexcolor0xffff00,draw=hexcolor0x000000,line width=0.8 pt]
\tikzstyle{wn}=[circle,fill=hexcolor0xffffff,draw=hexcolor0x000000,line width=0.8 pt]
\tikzstyle{wnthick}=[circle,fill=hexcolor0xffffff,draw=hexcolor0x000000,line width=2.500]

\tikzstyle{simple}=[-,draw=hexcolor0x000000,line width=2.000]
\tikzstyle{arrow}=[-,draw=hexcolor0x000000,postaction={decorate},decoration={markings,mark=at position .5 with {\arrow{>}}},line width=2.000]
\tikzstyle{tick}=[-,draw=hexcolor0x000000,postaction={decorate},decoration={markings,mark=at position .5 with {\draw (0,-0.1) -- (0,0.1);}},line width=2.000]
\tikzstyle{halfthickness}=[-,draw=hexcolor0x000000,line width=0.500]
\tikzstyle{thick}=[-,draw=hexcolor0x000000,line width=2.500]
\tikzstyle{thicker}=[-,draw=hexcolor0x000000,line width=4.000]

\tikzstyle{env}=[copoint,regular polygon rotate=0,minimum width=0.2cm, fill=black]

\tikzstyle{probs}=[shape=semicircle,fill=white,draw=black,shape border rotate=180,minimum width=1.2cm]

\tikzstyle{every picture}=[baseline=-0.25em,scale=0.5]
\tikzstyle{dotpic}=[] 
\tikzstyle{diredges}=[every to/.style={diredge}]
\tikzstyle{math matrix}=[matrix of math nodes,left delimiter=(,right delimiter=),inner sep=2pt,column sep=1em,row sep=0.5em,nodes={inner sep=0pt},text height=1.5ex, text depth=0.25ex]

\tikzstyle{inline text}=[text height=1.5ex, text depth=0.25ex,yshift=0.5mm]
\tikzstyle{label}=[font=\footnotesize,text height=1.5ex, text depth=0.25ex,yshift=0.5mm]
\tikzstyle{left label}=[label,anchor=east,xshift=1.5mm]
\tikzstyle{right label}=[label,anchor=west,xshift=-1.5mm]

\tikzstyle{braceedge}=[decorate,decoration={brace,amplitude=2mm,raise=-1mm}]
\tikzstyle{small braceedge}=[decorate,decoration={brace,amplitude=1mm,raise=-1mm}]

\tikzstyle{doubled}=[line width=1.6pt] 
\tikzstyle{boldedge}=[doubled,shorten <=-0.17mm,shorten >=-0.17mm]
\tikzstyle{boldedgegray}=[doubled,gray,shorten <=-0.17mm,shorten >=-0.17mm]

\tikzstyle{semidoubled}=[line width=1.4pt] 
\tikzstyle{semiboldedgegray}=[semidoubled,gray,shorten <=-0.17mm,shorten >=-0.17mm]

\tikzstyle{boldedgedashed}=[very thick,dashed,shorten <=-0.17mm,shorten >=-0.17mm]
\tikzstyle{vboldedgedashed}=[doubled,dashed,shorten <=-0.17mm,shorten >=-0.17mm]
\tikzstyle{left hook arrow}=[left hook-latex]
\tikzstyle{right hook arrow}=[right hook-latex]
\tikzstyle{sembracket}=[line width=0.5pt,shorten <=-0.07mm,shorten >=-0.07mm]

\tikzstyle{causal edge}=[->,thick,gray]
\tikzstyle{causal nondir}=[thick,gray]
\tikzstyle{timeline}=[thick,gray, dashed]

\tikzstyle{cedge}=[<->,thick,gray!70!white]

\tikzstyle{empty diagram}=[draw=gray!40!white,dashed,shape=rectangle,minimum width=1cm,minimum height=1cm]
\tikzstyle{empty diagram small}=[draw=gray!50!white,dashed,shape=rectangle,minimum width=0.6cm,minimum height=0.5cm]

\tikzstyle{dot}=[inner sep=0mm,minimum width=2mm,minimum height=2mm,draw,shape=circle]
\tikzstyle{ddot}=[inner sep=0mm, doubled, minimum width=2.5mm,minimum height=2.5mm,draw,shape=circle]

\tikzstyle{black dot}=[dot,fill=black]
\tikzstyle{white dot}=[dot,fill=white,,text depth=-0.2mm]
\tikzstyle{green dot}=[white dot] 
\tikzstyle{gray dot}=[dot,fill=gray!40!white,,text depth=-0.2mm]
\tikzstyle{red dot}=[gray dot] 

\tikzstyle{black ddot}=[ddot,fill=black]
\tikzstyle{white ddot}=[ddot,fill=white]
\tikzstyle{gray ddot}=[ddot,fill=gray!40!white]

\tikzstyle{gray edge}=[gray!40!white]

\tikzstyle{small dot}=[inner sep=0.5mm,minimum width=0pt,minimum height=0pt,draw,shape=circle]

\tikzstyle{small black dot}=[small dot,fill=black]
\tikzstyle{small white dot}=[small dot,fill=white]
\tikzstyle{small gray dot}=[small dot,fill=gray!40!white]

\tikzstyle{causal dot}=[inner sep=0.4mm,minimum width=0pt,minimum height=0pt,draw=white,shape=circle,fill=gray!40!white]

\tikzstyle{phase dimensions}=[minimum size=5mm,font=\footnotesize,rectangle,rounded corners=2.5mm,inner sep=0.2mm,outer sep=-2mm]
\tikzstyle{dphase dimensions}=[minimum size=5mm,font=\footnotesize,rectangle,rounded corners=2.5mm,inner sep=0.2mm,outer sep=-2mm]

\tikzstyle{white phase dot}=[dot,fill=white,phase dimensions]
\tikzstyle{white phase ddot}=[ddot,fill=white,dphase dimensions]

\tikzstyle{white rect ddot}=[draw=black,fill=white,doubled,minimum size=5mm,font=\footnotesize,rectangle,rounded corners=2.5mm,inner sep=0.2mm]
\tikzstyle{gray rect ddot}=[draw=black,fill=gray!40!white,doubled,minimum size=6mm,font=\footnotesize,rectangle,rounded corners=3mm]

\tikzstyle{gray phase dot}=[dot,fill=gray!40!white,phase dimensions]
\tikzstyle{gray phase ddot}=[ddot,fill=gray!40!white,dphase dimensions]
\tikzstyle{grey phase dot}=[gray phase dot]
\tikzstyle{grey phase ddot}=[gray phase ddot]

\tikzstyle{small phase dimensions}=[minimum size=4mm,font=\tiny,rectangle,rounded corners=2mm,inner sep=0.2mm,outer sep=-2mm]
\tikzstyle{small dphase dimensions}=[minimum size=4mm,font=\tiny,rectangle,rounded corners=2mm,inner sep=0.2mm,outer sep=-2mm]

\tikzstyle{small gray phase dot}=[dot,fill=gray!40!white,small phase dimensions]
\tikzstyle{small gray phase ddot}=[ddot,fill=gray!40!white,small dphase dimensions]

\tikzstyle{small map}=[draw,shape=rectangle,minimum height=4mm,minimum width=4mm,fill=white]

\tikzstyle{cnot}=[fill=white,shape=circle,inner sep=-1.4pt]

\tikzstyle{asym hadamard}=[fill=white,draw,shape=NEbox,inner sep=0.6mm,font=\footnotesize,minimum height=4mm]
\tikzstyle{asym hadamard conj}=[fill=white,draw,shape=NWbox,inner sep=0.6mm,font=\footnotesize,minimum height=4mm]
\tikzstyle{asym hadamard dag}=[fill=white,draw,shape=SEbox,inner sep=0.6mm,font=\footnotesize,minimum height=4mm]

\tikzstyle{hadamard}=[fill=white,draw,inner sep=0.6mm,font=\footnotesize,minimum height=4mm,minimum width=4mm]
\tikzstyle{small hadamard}=[fill=white,draw,inner sep=0.6mm,minimum height=1.5mm,minimum width=1.5mm]
\tikzstyle{dhadamard}=[hadamard,doubled]
\tikzstyle{small dhadamard}=[small hadamard,doubled]
\tikzstyle{small dhadamard rotate}=[small hadamard,doubled,rotate=45]
\tikzstyle{antipode}=[white dot,inner sep=0.3mm,font=\footnotesize]

\tikzstyle{scalar}=[diamond,draw,inner sep=0.5pt,font=\small]
\tikzstyle{dscalar}=[diamond,doubled, draw,inner sep=0.5pt,font=\small]

\tikzstyle{small box}=[rectangle,inline text,fill=white,draw,minimum height=5mm,yshift=-0.5mm,minimum width=5mm,font=\small]
\tikzstyle{small gray box}=[small box,fill=gray!30]
\tikzstyle{medium box}=[rectangle,inline text,fill=white,draw,minimum height=5mm,yshift=-0.5mm,minimum width=10mm,font=\small]
\tikzstyle{square box}=[small box] 
\tikzstyle{medium gray box}=[small box,fill=gray!30]
\tikzstyle{semilarge box}=[rectangle,inline text,fill=white,draw,minimum height=5mm,yshift=-0.5mm,minimum width=12.5mm,font=\small]
\tikzstyle{large box}=[rectangle,inline text,fill=white,draw,minimum height=5mm,yshift=-0.5mm,minimum width=15mm,font=\small]
\tikzstyle{large gray box}=[small box,fill=gray!30]

\tikzstyle{Bayes box}=[rectangle,fill=black,draw, minimum height=3mm, minimum width=3mm]

\tikzstyle{gray square point}=[small box,fill=gray!50]

\tikzstyle{dphase box white}=[dhadamard]
\tikzstyle{dphase box gray}=[dhadamard,fill=gray!50!white]

\tikzstyle{point}=[regular polygon,regular polygon sides=3,draw,scale=0.75,inner sep=-0.5pt,minimum width=9mm,fill=white,regular polygon rotate=180]
\tikzstyle{copoint}=[regular polygon,regular polygon sides=3,draw,scale=0.75,inner sep=-0.5pt,minimum width=9mm,fill=white]
\tikzstyle{dpoint}=[point,doubled]
\tikzstyle{dcopoint}=[copoint,doubled]

\tikzstyle{wide copoint}=[fill=white,draw,shape=isosceles triangle,shape border rotate=90,isosceles triangle stretches=true,inner sep=0pt,minimum width=1.5cm,minimum height=6.12mm]
\tikzstyle{wide point}=[fill=white,draw,shape=isosceles triangle,shape border rotate=-90,isosceles triangle stretches=true,inner sep=0pt,minimum width=1.5cm,minimum height=6.12mm,yshift=-0.0mm]
\tikzstyle{wide point plus}=[fill=white,draw,shape=isosceles triangle,shape border rotate=-90,isosceles triangle stretches=true,inner sep=0pt,minimum width=1.74cm,minimum height=7mm,yshift=-0.0mm]

\tikzstyle{wide dpoint}=[fill=white,doubled,draw,shape=isosceles triangle,shape border rotate=-90,isosceles triangle stretches=true,inner sep=0pt,minimum width=1.5cm,minimum height=6.12mm,yshift=-0.0mm]
\tikzstyle{wide dcopoint}=[fill=white,doubled,draw,shape=isosceles triangle,shape border rotate=90,isosceles triangle stretches=true,inner sep=0pt,minimum width=1.5cm,minimum height=6.12mm,yshift=-0.0mm]

\tikzstyle{tinypoint}=[regular polygon,regular polygon sides=3,draw,scale=0.55,inner sep=-0.15pt,minimum width=6mm,fill=white,regular polygon rotate=180]

\tikzstyle{white point}=[point]
\tikzstyle{white dpoint}=[dpoint]
\tikzstyle{green point}=[white point] 
\tikzstyle{white copoint}=[copoint]
\tikzstyle{gray point}=[point,fill=gray!40!white]
\tikzstyle{gray dpoint}=[gray point,doubled]
\tikzstyle{red point}=[gray point] 
\tikzstyle{gray copoint}=[copoint,fill=gray!40!white]
\tikzstyle{gray dcopoint}=[gray copoint,doubled]

\tikzstyle{white point guide}=[regular polygon,regular polygon sides=3,font=\scriptsize,draw,scale=0.65,inner sep=-0.5pt,minimum width=9mm,fill=white,regular polygon rotate=180]

\tikzstyle{black point}=[point,fill=black,font=\color{white}]
\tikzstyle{black copoint}=[copoint,fill=black,font=\color{white}]

\tikzstyle{tiny gray point}=[tinypoint,fill=gray!40!white]

\tikzstyle{diredge}=[->]
\tikzstyle{ddiredge}=[<->]
\tikzstyle{rdiredge}=[<-]
\tikzstyle{thickdiredge}=[->, very thick]
\tikzstyle{pointer edge}=[->,very thick,gray]
\tikzstyle{pointer edge part}=[very thick,gray]
\tikzstyle{dashed edge}=[dashed]
\tikzstyle{thick dashed edge}=[very thick,dashed]
\tikzstyle{thick gray dashed edge}=[thick dashed edge,gray!40]
\tikzstyle{thick map edge}=[very thick,|->]

\makeatletter
\newcommand{\boxshape}[3]{%
\pgfdeclareshape{#1}{
\inheritsavedanchors[from=rectangle] 
\inheritanchorborder[from=rectangle]
\inheritanchor[from=rectangle]{center}
\inheritanchor[from=rectangle]{north}
\inheritanchor[from=rectangle]{south}
\inheritanchor[from=rectangle]{west}
\inheritanchor[from=rectangle]{east}
\backgroundpath{
\southwest \pgf@xa=\pgf@x \pgf@ya=\pgf@y
\northeast \pgf@xb=\pgf@x \pgf@yb=\pgf@y

\@tempdima=#2
\@tempdimb=#3

\pgfpathmoveto{\pgfpoint{\pgf@xa - 5pt + \@tempdima}{\pgf@ya}}
\pgfpathlineto{\pgfpoint{\pgf@xa - 5pt - \@tempdima}{\pgf@yb}}
\pgfpathlineto{\pgfpoint{\pgf@xb + 5pt + \@tempdimb}{\pgf@yb}}
\pgfpathlineto{\pgfpoint{\pgf@xb + 5pt - \@tempdimb}{\pgf@ya}}
\pgfpathlineto{\pgfpoint{\pgf@xa - 5pt + \@tempdima}{\pgf@ya}}
\pgfpathclose
}
}}

\boxshape{NEbox}{0pt}{5pt}
\boxshape{SEbox}{0pt}{-5pt}
\boxshape{NWbox}{5pt}{0pt}
\boxshape{SWbox}{-5pt}{0pt}
\boxshape{EBox}{-3pt}{3pt}
\boxshape{WBox}{3pt}{-3pt}
\makeatother

\tikzstyle{cloud}=[shape=cloud,draw,minimum width=1.5cm,minimum height=1.5cm]

\tikzstyle{map}=[draw,shape=NEbox,inner sep=2pt,minimum height=6mm,fill=white]
\tikzstyle{dashedmap}=[draw,dashed,shape=NEbox,inner sep=2pt,minimum height=6mm,fill=white]
\tikzstyle{mapdag}=[draw,shape=SEbox,inner sep=2pt,minimum height=6mm,fill=white]
\tikzstyle{mapadj}=[draw,shape=SEbox,inner sep=2pt,minimum height=6mm,fill=white]
\tikzstyle{maptrans}=[draw,shape=SWbox,inner sep=2pt,minimum height=6mm,fill=white]
\tikzstyle{mapconj}=[draw,shape=NWbox,inner sep=2pt,minimum height=6mm,fill=white]

\tikzstyle{medium map}=[draw,shape=NEbox,inner sep=2pt,minimum height=6mm,fill=white,minimum width=7mm]
\tikzstyle{medium map dag}=[draw,shape=SEbox,inner sep=2pt,minimum height=6mm,fill=white,minimum width=7mm]
\tikzstyle{medium map adj}=[draw,shape=SEbox,inner sep=2pt,minimum height=6mm,fill=white,minimum width=7mm]
\tikzstyle{medium map trans}=[draw,shape=SWbox,inner sep=2pt,minimum height=6mm,fill=white,minimum width=7mm]
\tikzstyle{medium map conj}=[draw,shape=NWbox,inner sep=2pt,minimum height=6mm,fill=white,minimum width=7mm]
\tikzstyle{semilarge map}=[draw,shape=NEbox,inner sep=2pt,minimum height=6mm,fill=white,minimum width=9.5mm]
\tikzstyle{semilarge map trans}=[draw,shape=SWbox,inner sep=2pt,minimum height=6mm,fill=white,minimum width=9.5mm]
\tikzstyle{semilarge map adj}=[draw,shape=SEbox,inner sep=2pt,minimum height=6mm,fill=white,minimum width=9.5mm]
\tikzstyle{semilarge map dag}=[draw,shape=SEbox,inner sep=2pt,minimum height=6mm,fill=white,minimum width=9.5mm]
\tikzstyle{semilarge map conj}=[draw,shape=NWbox,inner sep=2pt,minimum height=6mm,fill=white,minimum width=9.5mm]
\tikzstyle{large map}=[draw,shape=NEbox,inner sep=2pt,minimum height=6mm,fill=white,minimum width=12mm]
\tikzstyle{large map conj}=[draw,shape=NWbox,inner sep=2pt,minimum height=6mm,fill=white,minimum width=12mm]
\tikzstyle{very large map}=[draw,shape=NEbox,inner sep=2pt,minimum height=6mm,fill=white,minimum width=17mm]

\tikzstyle{medium dmap}=[draw,doubled,shape=NEbox,inner sep=2pt,minimum height=6mm,fill=white,minimum width=7mm]
\tikzstyle{medium dmap dag}=[draw,doubled,shape=SEbox,inner sep=2pt,minimum height=6mm,fill=white,minimum width=7mm]
\tikzstyle{medium dmap adj}=[draw,doubled,shape=SEbox,inner sep=2pt,minimum height=6mm,fill=white,minimum width=7mm]
\tikzstyle{medium dmap trans}=[draw,doubled,shape=SWbox,inner sep=2pt,minimum height=6mm,fill=white,minimum width=7mm]
\tikzstyle{medium dmap conj}=[draw,doubled,shape=NWbox,inner sep=2pt,minimum height=6mm,fill=white,minimum width=7mm]
\tikzstyle{semilarge dmap}=[draw,doubled,shape=NEbox,inner sep=2pt,minimum height=6mm,fill=white,minimum width=9.5mm]
\tikzstyle{semilarge dmap trans}=[draw,doubled,shape=SWbox,inner sep=2pt,minimum height=6mm,fill=white,minimum width=9.5mm]
\tikzstyle{semilarge dmap adj}=[draw,doubled,shape=SEbox,inner sep=2pt,minimum height=6mm,fill=white,minimum width=9.5mm]
\tikzstyle{semilarge dmap dag}=[draw,doubled,shape=SEbox,inner sep=2pt,minimum height=6mm,fill=white,minimum width=9.5mm]
\tikzstyle{semilarge dmap conj}=[draw,doubled,shape=NWbox,inner sep=2pt,minimum height=6mm,fill=white,minimum width=9.5mm]
\tikzstyle{large dmap}=[draw,doubled,shape=NEbox,inner sep=2pt,minimum height=6mm,fill=white,minimum width=12mm]
\tikzstyle{large dmap conj}=[draw,doubled,shape=NWbox,inner sep=2pt,minimum height=6mm,fill=white,minimum width=12mm]
\tikzstyle{large dmap trans}=[draw,doubled,shape=SWbox,inner sep=2pt,minimum height=6mm,fill=white,minimum width=12mm]
\tikzstyle{large dmap adj}=[draw,doubled,shape=SEbox,inner sep=2pt,minimum height=6mm,fill=white,minimum width=12mm]
\tikzstyle{large dmap dag}=[draw,doubled,shape=SEbox,inner sep=2pt,minimum height=6mm,fill=white,minimum width=12mm]
\tikzstyle{very large dmap}=[draw,doubled,shape=NEbox,inner sep=2pt,minimum height=6mm,fill=white,minimum width=19.5mm]

\tikzstyle{muxbox}=[draw,shape=rectangle,minimum height=3mm,minimum width=3mm,fill=white]
\tikzstyle{dmuxbox}=[muxbox,doubled]

\tikzstyle{box}=[draw,shape=rectangle,inner sep=2pt,minimum height=6mm,minimum width=6mm,fill=white]
\tikzstyle{dbox}=[draw,doubled,shape=rectangle,inner sep=2pt,minimum height=6mm,minimum width=6mm,fill=white]
\tikzstyle{dmap}=[draw,doubled,shape=NEbox,inner sep=2pt,minimum height=6mm,fill=white]
\tikzstyle{dmapdag}=[draw,doubled,shape=SEbox,inner sep=2pt,minimum height=6mm,fill=white]
\tikzstyle{dmapadj}=[draw,doubled,shape=SEbox,inner sep=2pt,minimum height=6mm,fill=white]
\tikzstyle{dmaptrans}=[draw,doubled,shape=SWbox,inner sep=2pt,minimum height=6mm,fill=white]
\tikzstyle{dmapconj}=[draw,doubled,shape=NWbox,inner sep=2pt,minimum height=6mm,fill=white]

\tikzstyle{ddmap}=[draw,doubled,dashed,shape=NEbox,inner sep=2pt,minimum height=6mm,fill=white]
\tikzstyle{ddmapdag}=[draw,doubled,dashed,shape=SEbox,inner sep=2pt,minimum height=6mm,fill=white]
\tikzstyle{ddmapadj}=[draw,doubled,dashed,shape=SEbox,inner sep=2pt,minimum height=6mm,fill=white]
\tikzstyle{ddmaptrans}=[draw,doubled,dashed,shape=SWbox,inner sep=2pt,minimum height=6mm,fill=white]
\tikzstyle{ddmapconj}=[draw,doubled,dashed,shape=NWbox,inner sep=2pt,minimum height=6mm,fill=white]

\boxshape{sNEbox}{0pt}{3pt}
\boxshape{sSEbox}{0pt}{-3pt}
\boxshape{sNWbox}{3pt}{0pt}
\boxshape{sSWbox}{-3pt}{0pt}
\tikzstyle{smap}=[draw,shape=sNEbox,fill=white]
\tikzstyle{smapdag}=[draw,shape=sSEbox,fill=white]
\tikzstyle{smapadj}=[draw,shape=sSEbox,fill=white]
\tikzstyle{smaptrans}=[draw,shape=sSWbox,fill=white]
\tikzstyle{smapconj}=[draw,shape=sNWbox,fill=white]

\tikzstyle{dsmap}=[draw,dashed,shape=sNEbox,fill=white]
\tikzstyle{dsmapdag}=[draw,dashed,shape=sSEbox,fill=white]
\tikzstyle{dsmaptrans}=[draw,dashed,shape=sSWbox,fill=white]
\tikzstyle{dsmapconj}=[draw,dashed,shape=sNWbox,fill=white]

\boxshape{mNEbox}{0pt}{10pt}
\boxshape{mSEbox}{0pt}{-10pt}
\boxshape{mNWbox}{10pt}{0pt}
\boxshape{mSWbox}{-10pt}{0pt}
\tikzstyle{mmap}=[draw,shape=mNEbox]
\tikzstyle{mmapdag}=[draw,shape=mSEbox]
\tikzstyle{mmaptrans}=[draw,shape=mSWbox]
\tikzstyle{mmapconj}=[draw,shape=mNWbox]

\tikzstyle{mmapgray}=[draw,fill=gray!40!white,shape=mNEbox]
\tikzstyle{smapgray}=[draw,fill=gray!40!white,shape=sNEbox]

\makeatletter
\pgfdeclareshape{cornerpoint}{
\inheritsavedanchors[from=rectangle] 
\inheritanchorborder[from=rectangle]
\inheritanchor[from=rectangle]{center}
\inheritanchor[from=rectangle]{north}
\inheritanchor[from=rectangle]{south}
\inheritanchor[from=rectangle]{west}
\inheritanchor[from=rectangle]{east}
\backgroundpath{
\southwest \pgf@xa=\pgf@x \pgf@ya=\pgf@y
\northeast \pgf@xb=\pgf@x \pgf@yb=\pgf@y

\pgfmathsetmacro{\pgf@shorten@left}{\pgfkeysvalueof{/tikz/shorten left}}
\pgfmathsetmacro{\pgf@shorten@right}{\pgfkeysvalueof{/tikz/shorten right}}

\pgfpathmoveto{\pgfpoint{0.5 * (\pgf@xa + \pgf@xb)}{\pgf@ya - 5pt}}
\pgfpathlineto{\pgfpoint{\pgf@xa - 8pt + \pgf@shorten@left}{\pgf@yb - 1.5 * \pgf@shorten@left}}
\pgfpathlineto{\pgfpoint{\pgf@xa - 8pt + \pgf@shorten@left}{\pgf@yb}}
\pgfpathlineto{\pgfpoint{\pgf@xb + 8pt - \pgf@shorten@right}{\pgf@yb}}
\pgfpathlineto{\pgfpoint{\pgf@xb + 8pt - \pgf@shorten@right}{\pgf@yb - 1.5 * \pgf@shorten@right}}
\pgfpathclose
}
}

\pgfdeclareshape{cornercopoint}{
\inheritsavedanchors[from=rectangle] 
\inheritanchorborder[from=rectangle]
\inheritanchor[from=rectangle]{center}
\inheritanchor[from=rectangle]{north}
\inheritanchor[from=rectangle]{south}
\inheritanchor[from=rectangle]{west}
\inheritanchor[from=rectangle]{east}
\backgroundpath{
\southwest \pgf@xa=\pgf@x \pgf@ya=\pgf@y
\northeast \pgf@xb=\pgf@x \pgf@yb=\pgf@y

\pgfmathsetmacro{\pgf@shorten@left}{\pgfkeysvalueof{/tikz/shorten left}}
\pgfmathsetmacro{\pgf@shorten@right}{\pgfkeysvalueof{/tikz/shorten right}}

\pgfpathmoveto{\pgfpoint{0.5 * (\pgf@xa + \pgf@xb)}{\pgf@yb + 5pt}}
\pgfpathlineto{\pgfpoint{\pgf@xa - 8pt + \pgf@shorten@left}{\pgf@ya + 1.5 * \pgf@shorten@left}}
\pgfpathlineto{\pgfpoint{\pgf@xa - 8pt + \pgf@shorten@left}{\pgf@ya}}
\pgfpathlineto{\pgfpoint{\pgf@xb + 8pt - \pgf@shorten@right}{\pgf@ya}}
\pgfpathlineto{\pgfpoint{\pgf@xb + 8pt - \pgf@shorten@right}{\pgf@ya + 1.5 * \pgf@shorten@right}}
\pgfpathclose
}
}

\makeatother

\pgfkeyssetvalue{/tikz/shorten left}{0pt}
\pgfkeyssetvalue{/tikz/shorten right}{0pt}

\tikzstyle{kpoint common}=[draw,fill=white,inner sep=1pt,minimum height=4mm]
\tikzstyle{kpoint}=[shape=cornerpoint,shorten left=5pt,kpoint common]
\tikzstyle{kpoint adjoint}=[shape=cornercopoint,shorten left=5pt,kpoint common]
\tikzstyle{kpoint conjugate}=[shape=cornerpoint,shorten right=5pt,kpoint common]
\tikzstyle{kpoint transpose}=[shape=cornercopoint,shorten right=5pt,kpoint common]
\tikzstyle{kpoint symm}=[shape=cornerpoint,shorten left=5pt,shorten right=5pt,kpoint common]

\tikzstyle{black kpoint}=[shape=cornerpoint,shorten left=5pt,kpoint common,fill=black,font=\color{white}]
\tikzstyle{black kpoint adjoint}=[shape=cornercopoint,shorten left=5pt,kpoint common,fill=black,font=\color{white}]
\tikzstyle{black kpointadj}=[shape=cornercopoint,shorten left=5pt,kpoint common,fill=black,font=\color{white}]

\tikzstyle{black dkpoint}=[shape=cornerpoint,shorten left=5pt,kpoint common,fill=black, doubled,font=\color{white}]
\tikzstyle{black dkpoint adjoint}=[shape=cornercopoint,shorten left=5pt,kpoint common,fill=black, doubled,font=\color{white}]
\tikzstyle{black dkpointadj}=[shape=cornercopoint,shorten left=5pt,kpoint common,fill=black, doubled,font=\color{white}]

\tikzstyle{kpointdag}=[kpoint adjoint]
\tikzstyle{kpointadj}=[kpoint adjoint]
\tikzstyle{kpointconj}=[kpoint conjugate]
\tikzstyle{kpointtrans}=[kpoint transpose]

\tikzstyle{big kpoint}=[kpoint, minimum width=1.2 cm, minimum height=8mm, inner sep=4pt, text depth=3mm]

\tikzstyle{wide kpoint}=[kpoint, minimum width=1 cm, inner sep=2pt]
\tikzstyle{wide kpointdag}=[kpointdag, minimum width=1 cm, inner sep=2pt]
\tikzstyle{wide kpointconj}=[kpointconj, minimum width=1 cm, inner sep=2pt]
\tikzstyle{wide kpointtrans}=[kpointtrans, minimum width=1 cm, inner sep=2pt]

\tikzstyle{gray kpoint}=[kpoint,fill=gray!50!white]
\tikzstyle{gray kpointdag}=[kpointdag,fill=gray!50!white]
\tikzstyle{gray kpointadj}=[kpointadj,fill=gray!50!white]
\tikzstyle{gray kpointconj}=[kpointconj,fill=gray!50!white]
\tikzstyle{gray kpointtrans}=[kpointtrans,fill=gray!50!white]

\tikzstyle{gray dkpoint}=[kpoint,fill=gray!50!white,doubled]
\tikzstyle{gray dkpointdag}=[kpointdag,fill=gray!50!white,doubled]
\tikzstyle{gray dkpointadj}=[kpointadj,fill=gray!50!white,doubled]
\tikzstyle{gray dkpointconj}=[kpointconj,fill=gray!50!white,doubled]
\tikzstyle{gray dkpointtrans}=[kpointtrans,fill=gray!50!white,doubled]

\tikzstyle{white label}=[draw,fill=white,rectangle,inner sep=0.7 mm]
\tikzstyle{gray label}=[draw,fill=gray!50!white,rectangle,inner sep=0.7 mm]
\tikzstyle{black label}=[draw,fill=black,rectangle,inner sep=0.7 mm]

\tikzstyle{dkpoint}=[kpoint,doubled]
\tikzstyle{wide dkpoint}=[wide kpoint,doubled]
\tikzstyle{dkpointdag}=[kpoint adjoint,doubled]
\tikzstyle{wide dkpointdag}=[wide kpointdag,doubled]
\tikzstyle{dkcopoint}=[kpoint adjoint,doubled]
\tikzstyle{dkpointadj}=[kpoint adjoint,doubled]
\tikzstyle{dkpointconj}=[kpoint conjugate,doubled]
\tikzstyle{dkpointtrans}=[kpoint transpose,doubled]

\tikzstyle{kscalar}=[kpoint common, shape=EBox, inner xsep=-1pt, inner ysep=3pt,font=\small]
\tikzstyle{kscalarconj}=[kpoint common, shape=WBox, inner xsep=-1pt, inner ysep=3pt,font=\small]

 \tikzstyle{upground}=[circuit ee IEC,thick,ground,rotate=90,scale=2.5]
 \tikzstyle{downground}=[circuit ee IEC,thick,ground,rotate=-90,scale=2.5]
 \tikzstyle{bigground}=[regular polygon,regular polygon sides=3,draw=gray,scale=0.50,inner sep=-0.5pt,minimum width=10mm,fill=gray]

\tikzstyle{arrs}=[-latex,font=\small,auto]
\tikzstyle{arrow plain}=[arrs]
\tikzstyle{arrow dashed}=[dashed,arrs]
\tikzstyle{arrow bold}=[very thick,arrs]
\tikzstyle{arrow hide}=[draw=white!0,-]
\tikzstyle{arrow reverse}=[latex-]
\tikzstyle{cdnode}=[]

\definecolor{hexcolor0xa9a9a9}{rgb}{0.663,0.663,0.663}
\tikzstyle{GrayLine}=[dashed,draw=hexcolor0xa9a9a9]
\tikzstyle{gray}=[dashed,draw=hexcolor0xa9a9a9]

\theoremstyle{definition}
\newtheorem{theorem}{Theorem}[section]
\newtheorem*{theorem*}{Theorem}
\newtheorem{corollary}[theorem]{Corollary}

\newtheorem{prop}[theorem]{Proposition}

\newtheorem{defn}[theorem]{Definition}

\newtheorem{example*}[theorem]{Example*}
\newtheorem{examples*}[theorem]{Examples*}

\newtheorem{remark*}[theorem]{Remark*}

\DeclareMathOperator{\Tr}{Tr}

\title{Double Dilation $\neq$ Double Mixing}
\author{Maaike Zwart and Bob Coecke \\ \footnotesize University of Oxford\\ \footnotesize \texttt{maaike.zwart@cs.ox.ac.uk - bob.coecke@cs.ox.ac.uk}}
\date{}

\begin{document}
\maketitle

\begin{abstract}
Density operators are one of the key ingredients of quantum theory. They can be constructed in two ways: via a convex sum of `doubled kets' (i.e.~mixing), and by tracing out part of a  `doubled' two-system ket (i.e.~dilation). Both constructions can be iterated, yielding new mathematical species that have already found applications outside physics. However, as we show in this paper, the iterated constructions no longer yield the same mathematical species. Hence, the constructions `mixing' and `dilation' themselves are by no means equivalent. Concretely, when applying the Choi-Jamiolkowski isomorphism to the second iteration, dilation produces arbitrary symmetric bipartite states, while mixing only yields the disentangled ones. All results are proven using diagrams, and hence they hold not only for quantum theory, but also for a much more general class of process theories.
\end{abstract}

\section{Introduction}

In 1932, von Neumann introduced special operators, now called density operators, to describe statistical \em mixtures \em of quantum states \cite{VonNeumann1955}. Unlike classical probability distributions, density operators are able to describe mixtures that involve a superposition of states, making them suitable for quantum theory. von Neumann noticed that these density operators also arise when part of a state describing a composite system is discarded, a.k.a.~\em dilation\em. So two conceptually different physical processes happen to yield the same mathematical species in the quantum formalism, in other words, density operators are two-faced.

Mathematically, the fact that these two faces are distinct shows in the corresponding constructions. Following \cite{Ashoush}, density operators representing statistical mixtures are constructed by first matching each vector (ket) in the mixture with its corresponding functional (bra), turning the vectors into operators. We call this \em doubling \em. Then, these operators are combined in a convex sum, forming the density operator:
\begin{align*}
  \left\{|\phi_i\rangle\right\}_i & \mapsto \left\{|\phi_i\rangle\langle\phi_i|\right\}_i \\
   & \mapsto
   \sum_i \ p_i\  |\phi_i\rangle\langle\phi_i|
\end{align*}
where all $p_i\geq 1$ and  $\sum_i p_i =1$.

Density operators originating from dilation are also constructed by doubling a vector: this time a vector in a space of form $A \otimes B$. Then, part of the resulting operator is traced out, yielding again a density operator:
\begin{align*}
  |\psi_{AB}\rangle & \mapsto  |\psi_{AB}\rangle\langle\psi_{AB}| \\
 & \mapsto Tr_B \left( |\psi_{AB}\rangle\langle\psi_{AB}| \right),
\end{align*}
where $|\psi_{AB}\rangle$ is a vector in Hilbert space $A \otimes B$. As both constructions yield density operators, one is tempted to think of these as equivalent, which indeed many physicists do.

Iterating these constructions yields new mathematical species that were called \emph{dual density operators} in \cite{Ashoush2015, Ashoush}. However, it turns out that the iterated versions of the constructions are no longer equivalent. In fact, as we will prove in sections \ref{section:notthesame} and \ref{section:subspace}, the dual density operators resulting from double mixing form a proper subspace of those resulting from double dilation. Note that we use the term `double' mixing/dilation to emphasise that this procedure involves another round of doubling, hence distinguishing from either (i) mixing further already mixed states, and (ii) discarding a second system after discarding a first, both of which of course still yield ordinary density operators.

Of course, the constructions can be iterated once more, and even more after that. Each iteration yields new mathematical species, and from the second iteration onwards, mixtures always form a proper
subspace of what is obtained by dilation. We illustrate the results of iterating both constructions any finite number of times in section \ref{section:generalisation}, expanding the work in \cite{Ashoush}, which already considered the general case for dilation but not for mixing.

A closer examination of the differences between the results of double mixing and double dilation reveals that double mixing always results in symmetric disentangled states, whereas doubly dilated states are highly symmetrical, but not necessarily disentangled. This difference hints at a possible classification of states that are either doubly mixed or doubly dilated, which would contribute to the characterisation of entangled states started by Horodecki\cite{Horodecki2007}. We make a start of this enterprise in sections \ref{section:characterisation_twicemix} and \ref{section:characterisation_twicedilation}.

In quantum theory, density operators provide enough structure to describe the currently known phenomena, so for physics there seems to be no direct use for the iterated constructions. The only  notable exception known to us is the study of the space of iterated dilated states as a generalised probabilistic theory  by Barnum and Barrett \cite{Barnum-Barrett}. However, recent developments in \em natural language processing \em (NLP) have found an interesting application for these generalised variants of density operators \cite{Ashoush2015, Ashoush}. Density operators first appeared in the NLP literature in \cite{blacoe2013quantum}, where they are mainly used to enlarge the parameter space. A conceptual grounding matching that of quantum theory is given in the work of Piedeleu, Kartsaklis et al.~\cite{Piedeleu2015}, where they are used as a model for ambiguous words, representing these words as a statistical mixture over their possible `pure'  meanings. Here, the overall setting was that of categorical compositional distributional (DisCoCat) models of meaning of Coecke, Sadrzadeh and Clark \cite{Coeckea}, which was itself also strongly inspired by quantum theory \cite{teleling}. 

A second application of density operators in NLP is found in the work of Balkir, Bankova et al.~\cite{Balkir2016, bankova2016graded, Balkr}, where lexical entailment is modelled by exploiting the fact that density operators can be partially ordered \cite{CoeckeMartin, Weteringen}. Naturally, this brought the need for a model that could accommodate ambiguity and lexical entailment simultaneously. It was to this end that Ashoush and Coecke started iterating the constructions of density operators. The results of this paper will hopefully contribute to further developing such models for NLP.

The structure of this paper is as follows. After a brief explanation of the graphical notation used in this paper, we recap the two constructions of density operators as described in \cite{Ashoush2015}, using both traditional Dirac bra-ket notation and diagrams. Then, in section \ref{section:notthesame}, we show that the results of iterating these constructions twice are no longer identical. The results are however strongly related: one being a subspace of the other, which we prove in section \ref{section:subspace}. Next, we characterise what results from double mixing as a special class of disentangled states in section \ref{section:characterisation_twicemix}, and reveal the symmetries caused by double dilation in section \ref{section:characterisation_twicedilation}. Finally, in section \ref{section:generalisation} we analyse the general case of applying both constructions any finite number of times, which emphasises the difference between them.

\section{Graphical notation}\label{section:diagram_notation}

Our proofs use a diagrammatic language designed for categorical quantum mechanics \cite{Kindergarten, Selinger2007, CPV, CPaqPav}, building further on Penrose's notation \cite{Penrose1971}.  We use this graphical notation because it greatly simplifies otherwise tedious proofs, abstracting away from unimportant details. It also has the advantage that the results are true in a more general setting than just Hilbert spaces. For the reader unfamiliar with this graphical notation we briefly summarise the main components that will feature in this paper. For an extensive introduction we refer to the textbook \cite{CoeckeBOOK} or the shorter  paper version \cite{Coecke2015, Coecke2016} which is also self-contained.

\subsection{The basics: boxes and wires}\label{section:basic_diagrams}

The diagrammatic language represents operators as boxes with inputs and outputs. Wires that go into or come out of these boxes represent the Hilbert spaces on which the operators act.
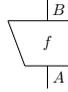
\[
  \scalebox{0.6}{\input{general_process_withtypes.tikz}}
\]
Operator composition is depicted by connecting the output of one box with the input of another, while tensoring operators is depicted by simply drawing the boxes next to each other.
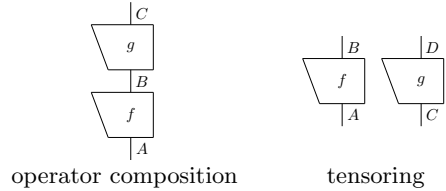
\begin{center}
  \begin{tabular}{ccc}
    \scalebox{0.6}{\input{general_process_composition_withtypes.tikz}} & \qquad & \scalebox{0.6}  {\input{general_process_tensor_withtypes.tikz}} \\
    \small operator composition & \qquad &\small tensoring\\
    \end{tabular}
\end{center}
Vectors (kets) and functionals (bras) respectively are boxes with no input(s) or no output(s), so that an operator acting on a vector boils down to connecting the vector and the operator. We can turn a vector into a functional via horizontal reflection, which then yields a diagrammatic representation of the inner-product.
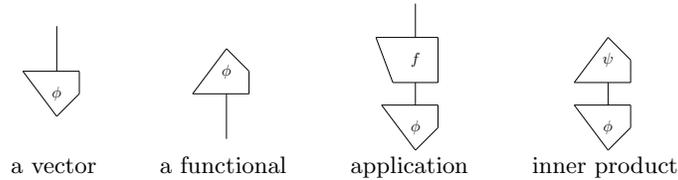
\begin{center}
  \begin{tabular}{ccccccc}
     \scalebox{0.6}{\input{general_state.tikz}} & \qquad & \scalebox{0.6}{\input{general_effect.tikz}} & \qquad & \scalebox{0.6}{\input{general_process_appliedto_state.tikz}} & \qquad & \scalebox{0.6}{\input{general_state_effect.tikz}} \\
     \small a vector & \qquad & \small a functional & \qquad & \small application & \qquad & \small inner product \\
    \end{tabular}
\end{center}
Inner products produce numbers, which are boxes with no inputs and no outputs. In this paper, we are not too concerned about numbers. Whenever two diagrams are equal to each other up to a non-zero number, we indicate this with $\approx$ instead of $=$.
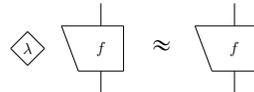
\[
\scalebox{0.6}{\input{general_process_with_number.tikz}} \;\; \approx \;\; \scalebox{0.6}{\input{general_process.tikz}}
\]

There is a reason for all boxes to be drawn asymmetrically: conjugates, adjoints and transposes are depicted as reflections and rotations of the original operator.
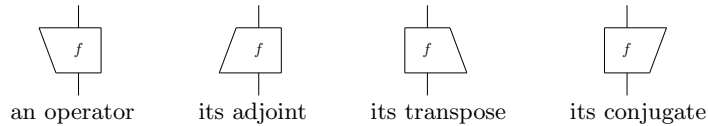
\begin{center}
  \begin{tabular}{ccccccc}
     \scalebox{0.6}{\input{general_process.tikz}} & \qquad &
     \scalebox{0.6}{\input{general_process_adjoint.tikz}} & \qquad & \scalebox{0.6} {\input{general_process_transpose.tikz}}  & \qquad & \scalebox{0.6}{\input{general_process_conjugate.tikz}}\\
     \small an operator & \qquad & \small its adjoint & \qquad & \small its transpose & \qquad & \small its conjugate \\
    \end{tabular}
\end{center}
where the latter is the result of both taking the adjoint and the transpose. Whenever the adjoint/conjugate of an operator is equal to the operator itself, it is called self-adjoint/self-conjugate. Sometimes, self-adjointness or self-conjugateness of an operator is clear from the corresponding diagram: it has vertical/horizontal reflection symmetries (e.g.~the diagrams below). We call such operators adjoint-symmetric/ conjugate-symmetric. In the literature, adjoint-symmetric operators are also called \em positive \em operators, given that, in linear algebra, that is what these diagrams correspond to.
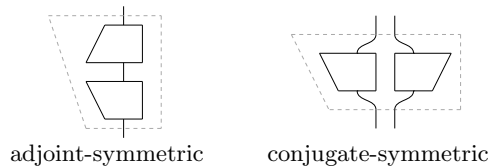
\begin{center}
\begin{tabular}{ccc}
  \scalebox{0.5}{\input{adjoint-symmetric.tikz}} & \qquad & \scalebox{0.5}{\input{conjugate-symmetric.tikz}} \\
  \small adjoint-symmetric & &\small conjugate-symmetric
\end{tabular}
\end{center}

Since the transpose operates between the dual spaces, so does the conjugate. We represent the dual space of $A$ by $\overline{A}$. Although we distinguish notationally between a space and its dual, we will treat them as being the same, which can be done by fixing a basis.

We use thick wires of type $\hat{A}$ as shorthand for two wires of respective types $A$ and $\overline{A}$. Operators can also be drawn as boxes with thick lines. Analogous to thick wires, this is shorthand for tensoring the operator and its conjugate.
\[
\scalebox{0.8}{\input{general_thick_wire.tikz}} \qquad\qquad \scalebox{0.6}{\input{general_thick_process.tikz}}
\]

\subsection{Spiders}\label{section:spiders}

Apart from boxes there are spiders (little circles), which are used as shorthand for various sums over basis vectors. Each leg of the spider is an incoming (bra) or outgoing (ket) basis vector.
\begin{center}
  \begin{tabular}{ccccc}
    \scalebox{0.6}{\input{spider_straight.tikz}} & \qquad & \scalebox{0.6}{\input{spider_copy.tikz}} & \qquad & \scalebox{0.6}{\input{general_spider.tikz}}\\
   \small identity spider & \qquad & \small copy spider & \qquad & \small general spider\\
  \end{tabular}
\end{center}
When two or more spiders are connected, they can fuse together. Conversely, a single spider can fission into multiple ones. Spider fusion and fission follows to the following rule:
\begin{center}
  \begin{tabular}{ccc}
    \scalebox{0.6}{\input{general_spider_defuse.tikz}} & $\leftrightarrow$ & \scalebox{0.6}{\input{general_spider_fused.tikz}} \\
      \multicolumn{3}{c}{\small ltr: fusion, rtl: fission} \\
  \end{tabular}
\end{center}
Spiders come in a family of three:
classical spiders, quantum spiders and bastard spiders. Classical spiders are the ones drawn above, simple thin-wired spiders. Quantum spiders are the thick-wired cousins of classical spiders. Just as with operators, a quantum spider is a classical spider paired with its conjugate. Quantum spiders follow the same fusion and fission rules as classical spiders.
\begin{center}
  \begin{tabular}{c}
    \scalebox{0.6}{\input{general_quantum_spider.tikz}} \\
   \small a quantum spider \\
  \end{tabular}
\end{center}
Bastard spiders are a mix of classical and quantum spiders. The spider itself is classical, but it can have both quantum and classical legs. Bastard spiders can fuse and fission with all three species. When a bastard spider fuses with a quantum spider, the resulting spider is always a bastard spider.
\begin{center}
  \begin{tabular}{c}
    \scalebox{0.6}{\input{general_spider_defuse_bastard-quantum.tikz}} \;\; = \;\; \scalebox{0.6}{\input{general_spider_fused_bastard.tikz}} \\
   \small a bastard spider and a quantum spider meet\\
  \end{tabular}
\end{center}
In some special cases, we omit the spider circle:
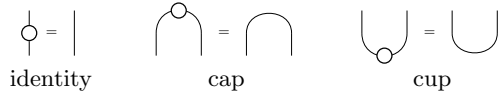
\begin{center}
  \begin{tabular}{ccccc}
    \scalebox{0.6}{\input{spider_identity.tikz}} & \qquad & \scalebox{0.6}{\input{spider_cap.tikz}} & \qquad & \scalebox{0.6}{\input{spider_cup.tikz}} \\
   \small identity & \qquad & \small cap & \qquad & \small cup \\
  \end{tabular}
\end{center}
Another special spider is the following bastard spider:
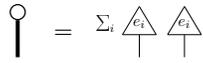
\[
\scalebox{0.6}{\begin{tikzpicture}
	\begin{pgfonlayer}{nodelayer}
		\node [style=none] (0) at (0, -1) {};
		\node [style=wn] (1) at (0, 1) {};
	\end{pgfonlayer}
	\begin{pgfonlayer}{edgelayer}
		\draw [style=thick] (1) to (0.center);
	\end{pgfonlayer}
\end{tikzpicture}} \;\; =  \;\; \scalebox{0.6}{\input{spider_discarding_sum.tikz}}
\]
It is just a cap in which the two inputs are seen as a quantum wire. It is often called discarding and has a corresponding notation:
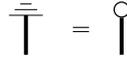
\[
\scalebox{0.6}{\input{discarding.tikz}} \;\; = \;\; \scalebox{0.6}{\begin{tikzpicture}
	\begin{pgfonlayer}{nodelayer}
		\node [style=none] (0) at (0, -1) {};
		\node [style=wn] (1) at (0, 1) {};
	\end{pgfonlayer}
	\begin{pgfonlayer}{edgelayer}
		\draw [style=thick] (1) to (0.center);
	\end{pgfonlayer}
\end{tikzpicture}}
\]

\subsection{Trace and changing sums into spiders}\label{section:changing_sum_into_spider}

We can now draw the trace of an operator quite elegantly. The trace of an operator is the sum over the elements on the diagonal of its matrix. In other words, if we choose an orthonormal basis, then the trace is the following sum (first diagram below). As $\langle e_i | e_i \rangle = 1$, the second diagram is equal to the first. The third diagram exploits the fact that $\langle e_i | e_j \rangle = \delta_{ij}$, so it is safe to change some of the indices. Now we can replace the sums in the diagram with spiders, yielding the fourth diagram. All the resulting spiders happen to be special spiders mentioned above (a cup, a cap and the identity), so we do not actually draw them. The last diagram is the standard diagram for the trace of an operator\cite{Selinger2007}.
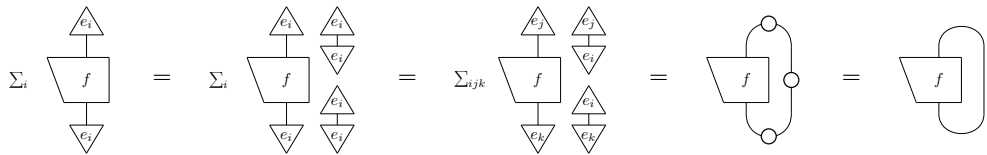
\[
\scalebox{0.6}{\input{general_process_trace_sum.tikz}} \;\;\;\; = \;\;\;\; \scalebox{0.6}{\input{general_process_trace_sum_additional_basisvectors.tikz}} \;\;\;\; = \;\;\;\;
\scalebox{0.6}{\input{general_process_trace_sum_additional_basisvectors_and_sum.tikz}} \;\;\;\; = \;\;\;\;
\scalebox{0.6}{\input{general_process_trace_spiders.tikz}} \;\;\;\; = \;\;\;\;
\scalebox{0.6}{\input{general_process_trace.tikz}}
\]

Trace is not the only diagram in which sums can be replaced by spiders. In fact, \em all \em finite sums appearing in diagrams can be replaced by spiders. In general, for diagrams with a sum like the one depicted below, the strategy is as follows.
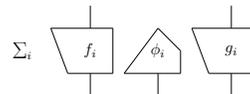
\[
\scalebox{0.6}{\input{replacing_sums_spiders_1.tikz}}
\]
First, choose a Hilbert space of dimension equal to the cardinality of the sum and choose an orthonormal basis for that Hilbert space. Since we have:
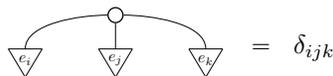
\[
\scalebox{0.6}{\input{replacing_sums_spiders_new.tikz}} \;\;  =  \;\; \delta_{ijk}
\]
the diagram can be rewritten as follows.
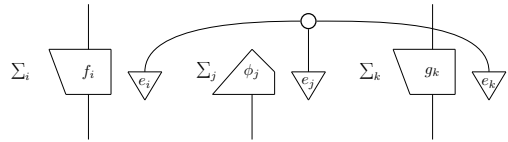
\[
\scalebox{0.6}{\input{replacing_sums_spiders_5.tikz}}
\]
Finally, we internalise the sums in boxes.
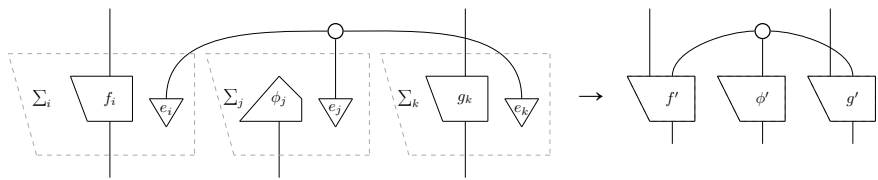
\[
\scalebox{0.6}{\input{replacing_sums_spiders_6.tikz}} \;\; \rightarrow \;\;
\scalebox{0.6}{\input{replacing_sums_spiders_8.tikz}}
\]
Replacing sums with spiders makes it easier to see the overall structure of the diagram: we now clearly see which boxes had correlated indices.

\subsection{Choi-Jamiolkowski isomorphism}\label{section:process_state_duality}

For notational convenience, rather than working with density operators we will work with their \em Choi-Jamiolkowski \em isomorphic vectors.  The Choi-Jamiolkowski isomorphism is the one-to-one correspondence between operators and vectors on tensor product spaces: when bending up the input wire from an operator,
it changes to a second output wire. Conversely, bending down one of two output wires turns it back into an operator.
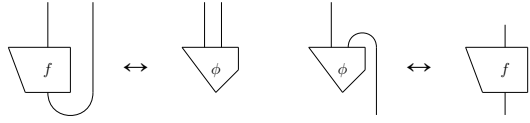
\begin{center}
  \begin{tabular}{ccc}
  \scalebox{0.6}{\input{general_process_bentup.tikz}}\;\; $\leftrightarrow$ \;\; \scalebox{0.6}{\input{general_state_bipartite.tikz}} & \qquad & \scalebox{0.6}{\input{general_state_bipartite_bentdown.tikz}}\;\;
    $\leftrightarrow$ \;\; \scalebox{0.6}{\input{general_process.tikz}}
  \end{tabular}
\end{center}
As a result of using this, each iterated application of the mixing and dilation constructions will only involve vectors instead of super-operators, super-super-operators and so on. In addition, partial traces turn into discarding:
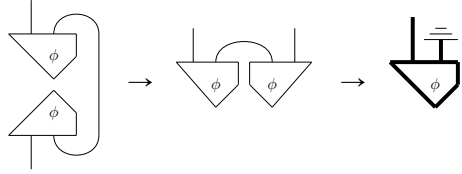
\[
\scalebox{0.6}{\input{trace_special.tikz}} \;\; \rightarrow \;\; \scalebox{0.6}{\input{trace_special_2.tikz}} \;\; \rightarrow \;\; \scalebox{0.6}{\input{trace_special_3.tikz}}
\]
The idea to use the Choi-Jamiolkowski isomorphism to write density operators as vectors comes from \cite{Selinger2007}.

\section{Double mixing and double dilation}\label{section:constuction}

We recap both mixing and dilation, and show how these constructions can be iterated.

\subsection{Mixing}\label{section:mixing}
Given a set of $n$ normalised vectors $|\phi_i\rangle$ in a finite-dimensional Hilbert space $A$, and a probability distribution $\{p_i\}^{i\leq n}$, we form the density operator representing the mixture of these vectors as follows:
\begin{align*}
  (\{ |\phi_i\rangle\},\{ p_i\}) \mapsto \rho = & \sum_{i=1}^{n} p_i |\phi_i\rangle \langle \phi_i|
\end{align*}
In the diagrammatic language:
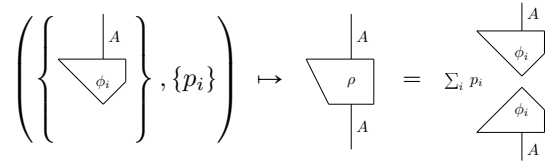
\[
\left(\left\{ \scalebox{0.6}{\input{phi_i_state.tikz}} \right\},\left\{ p_i \right\}\right)\;\; \mapsto \;\;\scalebox{0.6}{\input{rho_process.tikz}} \;\; = \;\; \scalebox{0.6}{\input{resultofmixing_process.tikz}}
\]
Alternatively, we could express $\rho$ as a vector in $A\otimes\overline{A}$. The benefit of obtaining a vector rather than an operator is that we get a construction that can be iterated, since it sends vectors to vectors:
\begin{align}\label{construction:sum}
(\{ |\phi_i\rangle\},\{ p_i\}) \mapsto ||\rho\rangle\rangle = & \sum_{i=1}^{n} p_i |\phi_i\rangle\overline{|\phi_i\rangle}
\end{align}
Here $\overline{|\phi_i\rangle}$ is the conjugate of $|\phi_i\rangle$, $|\phi_i\rangle\overline{|\phi_i\rangle}$ is shorthand for $|\phi_i\rangle \otimes \overline{|\phi_i\rangle}$ and notation $||\cdot\rangle\rangle$ is used to remind us that $\rho$ is a vector in Hilbert space $\hat{A} = A\otimes \overline{A}$ instead of $A$.

Diagrammatically, this construction translates as:
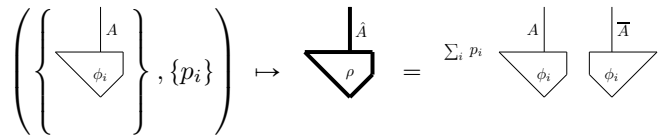
\[
\left(\left\{ \scalebox{0.6}{\input{phi_i_state.tikz}} \right\},\left\{ p_i \right\}\right)\;\; \mapsto \;\;\scalebox{0.6}{\input{rho_state.tikz}} \;\; = \;\; \scalebox{0.6}{\input{resultofmixing.tikz}}
\]

A second iteration of (\ref{construction:sum}) with $m$ density vectors $||\rho_k\rangle\rangle = \sum_{i=1}^{n_k} p_{ki} |\phi_{ki}\rangle \overline{|\phi_{ki}\rangle}$ and a probability distribution $\{r_k\}$ yields:
\begin{align}
  |||\Psi\rangle\rangle\rangle = & \sum_{k=1}^{m} r_k ||\rho_k\rangle\rangle \overline{||\rho_k\rangle\rangle} \nonumber
 \\
  = & \sum_{k=1}^{m} r_k \left( \sum_{i=1}^{n_k} p_{ki} |\phi_{ki}\rangle \overline{|\phi_{ki}\rangle}\right) \left( \overline{\sum_{j=1}^{n_k} p_{kj} |\phi_{kj}\rangle \overline{|\phi_{kj}\rangle}}\right)\nonumber \\
  = & \sum_{k=1}^{m}\sum_{i=1}^{n_k}\sum_{j=1}^{n_k} r_k p_{ki} p_{kj} |\phi_{ki}\rangle \overline{|\phi_{ki}\rangle}|\phi_{kj}\rangle \overline{|\phi_{kj}\rangle}\label{result:doublemix}
\end{align}
Where $|||\Psi\rangle\rangle\rangle$ is a vector in $A\otimes\overline{A}\otimes A\otimes\overline{A} = \hat{\hat{A}}$. The accompanying diagram is (only showing the result):
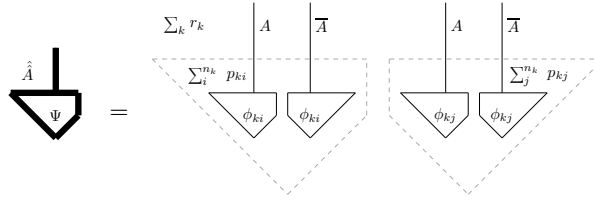
\[
\scalebox{0.6}{\input{Psi_state.tikz}} \;\; = \;\;\scalebox{0.6}{\input{resultofmixing_twice.tikz}}
\]
The dotted lines indicate the vectors $||\rho_k\rangle\rangle$ and $\overline{||\rho_k\rangle\rangle}$.

To make the diagram look prettier and to make it easier to compare to later results, we can hide the summations over $i$ and $j$ inside caps (wires $B$ in the diagram below) and summation over $k$ inside a four-legged spider (wire $C$), following the procedure `changing sums into spiders' explained in section \ref{section:changing_sum_into_spider}. The individual $\phi_{ki}$ will be no longer visible; they are absorbed into a general $\phi$:
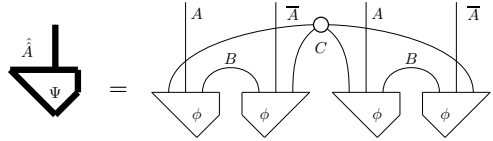
\[
\scalebox{0.6}{\input{Psi_state.tikz}} \;\; = \;\;\scalebox{0.6}{\input{doublemixstate_phi.tikz}}
\]
We call a vector resulting from twice applying construction (\ref{construction:sum}) \emph{doubly mixed}.

\subsection{Dilation}

On the other hand, if we have a vector $|\phi_{AB}\rangle$ in space $A\otimes B$, we can form the operator $|\phi_{AB}\rangle\langle\phi_{AB} |$ and trace out $B$:
\begin{align*}
 |\phi_{AB}\rangle \mapsto \rho' = & \Tr_B |\phi_{AB}\rangle\langle\phi_{AB}|
\end{align*}
Or as a diagram:
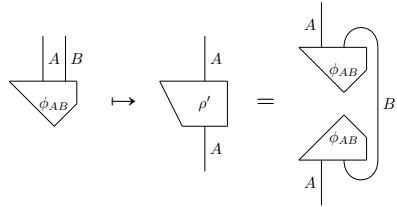
\[
\scalebox{0.6}{\input{phi_AB_state.tikz}} \;\; \mapsto\;\; \scalebox{0.6}{\input{rho_prime_process.tikz}}\;\; = \;\;\scalebox{0.6}{\input{Trace_B_Phi_AB.tikz}}
\]
If we would rather have a vector, this construction becomes:
\begin{align}\label{construction:reduce}
 |\phi_{AB}\rangle \mapsto ||\rho'\rangle\rangle = &  \sum_{i=0}^{\dim(B)-1}\langle e^B_i |\phi_{AB}\rangle \overline{\langle e^B_i |\phi_{AB}\rangle}
\end{align}
For some orthonormal basis $\{|e^B_i\rangle\}$ of $B$.
The corresponding diagram is:
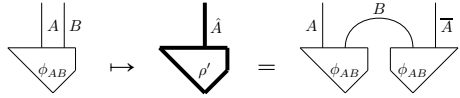
\[
\scalebox{0.6}{\input{phi_AB_state.tikz}} \;\; \mapsto\;\; \scalebox{0.6}{\input{rho_prime_state.tikz}}\;\; = \;\;\scalebox{0.6}{\input{reducedstate.tikz}}
\]

To iterate (\ref{construction:reduce}), we need the Hilbert space $A$ to be of form $A \otimes C$, so that after reducing a second time, the result is still a vector instead of a number. So suppose that our original vector was $|\phi_{ABC}\rangle$. Applying (\ref{construction:reduce}) using space $B$ then yields $||\rho'_{\hat{A}\hat{C}}\rangle\rangle$ in $A \otimes C \otimes \overline{C} \otimes \overline{A}$. Applying (\ref{construction:reduce}) again, now using space $\hat{C} = C \otimes \overline{C}$ results in:

\begin{align}
  |||\Psi'\rangle\rangle\rangle = & \sum_{k,l}\langle\langle e^{\hat{C}}_{k,l} ||\rho'_{\hat{A}\hat{C}}\rangle\rangle \overline{\langle\langle e^{\hat{C}}_{k,l} ||\rho'_{\hat{A}\hat{C}}\rangle\rangle} \nonumber \\
  = &  \sum_{k}\sum_{l}\left( \sum_{i}\langle e^C_k| \langle e^B_i|\phi_{ABC}\rangle \overline{\langle e^C_l| \langle e^B_i|\phi_{ABC}\rangle} \right) \left(\overline{\sum_{j}\langle e^C_k| \langle e^B_j|\phi_{ABC}\rangle \overline{\langle e^C_l| \langle e^B_j|\phi_{ABC}\rangle}} \right) \nonumber \\
  = & \sum_{k}\sum_{l}\sum_{i}\sum_{j}\langle e^C_k| \langle e^B_i|\phi_{ABC}\rangle \overline{\langle e^C_l| \langle e^B_i|\phi_{ABC}\rangle} \langle e^C_l| \langle e^B_j|\phi_{ABC}\rangle \overline{\langle e^C_k| \langle e^B_j|\phi_{ABC}\rangle} \label{result:doublereduce}
\end{align}
The corresponding diagram is:
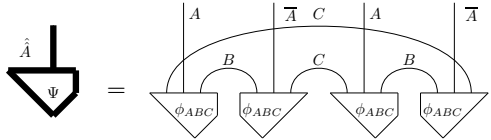
\[
\scalebox{0.6}{\input{Psi_state.tikz}} \;\; = \;\;\scalebox{0.6}{\input{dualdensitymatrix1_phi_ABC.tikz}}
\]
We call such vectors \emph{doubly dilated}.

\section{Counterexample to equivalence}\label{section:notthesame}

Comparing expressions (\ref{construction:sum}) and (\ref{result:doublemix}) to (\ref{construction:reduce}) and (\ref{result:doublereduce}), it is not at all obvious that mixing and dilation could be equivalent; indeed, they are not! Although it is possible to prove that both constructions give the same results when applied just a single time, this is no longer true when the constructions are iterated. As a counterexample, we give a vector resulting from double dilation that cannot be the result of double mixing.

\begin{theorem}\label{theorem_strictsubspace}
There exist vectors resulting from double dilation that cannot be written as vectors resulting from double mixing.
\end{theorem}

\begin{proof}
We will show that the following diagram on the left cannot be formed with either of the two diagrams on the right:
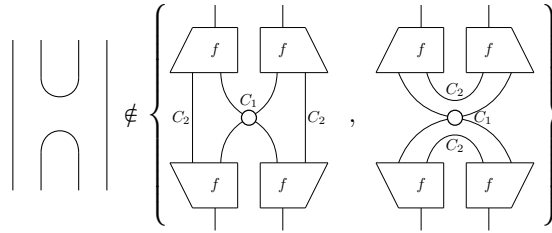
\[
\scalebox{1.0}{\input{notadoublemixture1.tikz}} \notin \left\{ \scalebox{0.6}{\input{doublemixmorphism2.tikz}}\;\;,\;\; \scalebox{0.6}{\input{doublemixmorphism2otherway.tikz}}  \right\}
\]
This will be enough to prove the theorem, as the diagram on the left is Choi-Jamiolkowski isomorphic to this vector resulting from double dilation:
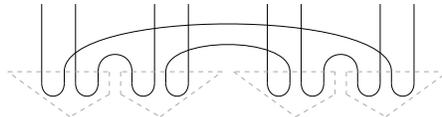
\begin{center}
  \begin{tabular}{c}
    \scalebox{0.6}{\input{counterexample_state.tikz}} \\
    \small Bend down the four leftmost outputs of the vector to get the given operator.
  \end{tabular}
\end{center}
while the two diagrams on the right are the only two possible operators resulting from the Choi-Jamiolkowski isomorphism applied to a vector resulting from double mixing (also see equations \ref{dm1} and \ref{dm2} in section \ref{section:characterisation_twicedilation} below).

First, we prove that the identity morphism cannot be of this form:
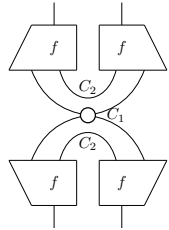
\[
\scalebox{0.6}{\input{doublemixmorphism2otherway.tikz}}
\]
To show this, suppose that the identity \emph{can} be written in that form:
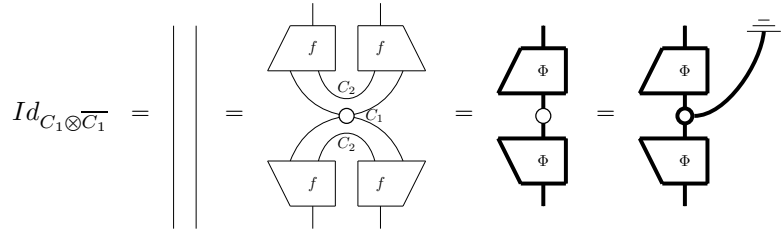
\begin{align*}
  Id_{C_1 \otimes \overline{C_1}}\;\; &= \;\; \scalebox{0.6}{\begin{tikzpicture}
	\begin{pgfonlayer}{nodelayer}
		\node [style=none] (0) at (-0.5000002, -5) {};
		\node [style=none] (1) at (-0.5, 4) {};
		\node [style=none] (2) at (0.5000002, 4) {};
		\node [style=none] (3) at (0.5000002, -5) {};
	\end{pgfonlayer}
	\begin{pgfonlayer}{edgelayer}
		\draw (2.center) to (3.center);
		\draw (1.center) to (0.center);
	\end{pgfonlayer}
\end{tikzpicture}} \;\; = \;\; \scalebox{0.6}{\input{doublemixmorphism2otherway.tikz}} \;\; = \;\;
  \scalebox{0.6}{\input{doublemixmorphism2otherway_dubbelgevouwen_spider.tikz}}\;\; = \;\; \scalebox{0.6}{\input{doublemixmorphism2otherway_dubbelgevouwen.tikz}}
\end{align*}
Here, the third equality is just rewriting the wires and boxes using thick lines (see section \ref{section:basic_diagrams} for explanation about thick wires and boxes). By doing this, the classical spider turns into a bastard spider (see section \ref{section:spiders} for explanation about spiders). In the last diagram, the bastard spider used fission to turn into a quantum spider and the special discarding spider.

We can now use the following theorem:
\begin{quotation}
``If a reduced operator $\Psi_{reduced}$ (an operator with one of its outputs discarded) is pure (can be written as a tensor product of some operator and its conjugate: $\Psi_{reduced} = \Phi \otimes \overline{\Phi}$), then the original operator can be written as a tensor product of that pure operator and a (possibly impure) vector: $\Psi = \Phi \otimes \overline{\Phi} \otimes ||\rho\rangle\rangle$'' \cite[Proposition 6.80]{CoeckeBOOK}.
\end{quotation}
Applying this theorem to the rightmost diagram above gives:
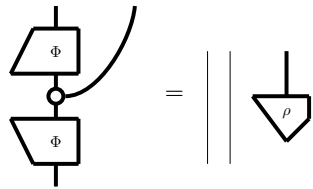
\[
\scalebox{0.6}{\input{doublemixmorphism2otherway_dubbelgevouwen_2.tikz}}\;\; = \;\; \scalebox{0.6}{\input{identity_and_state.tikz}}
\]
For some vector $\rho$. As $\rho$ is non-zero, there exists a basis vector $\langle e_i |$ such that the following is nonzero:
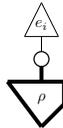
\[
\scalebox{0.6}{\input{rho_effect_i.tikz}}
\]
Combining this with the above:
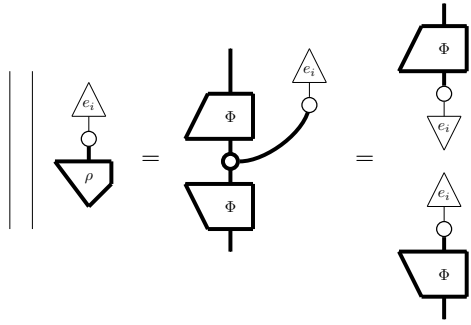
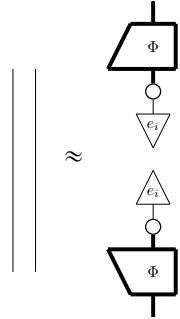
\begin{align*}
  \scalebox{0.6}{\input{identity_and_state_with_i.tikz}} \;\; & = \;\;\scalebox{0.6}{\input{doublemixmorphism2otherway_dubbelgevouwen_with_i.tikz}} \;\;  = \;\; \scalebox{0.6}{\input{doublemixmorphism2otherway_dubbelgevouwen_disconnect.tikz}}
\end{align*}
And hence:
\[
\scalebox{0.6}{\begin{tikzpicture}
	\begin{pgfonlayer}{nodelayer}
		\node [style=none] (0) at (-0.5000002, -5) {};
		\node [style=none] (1) at (-0.5, 4) {};
		\node [style=none] (2) at (0.5000002, 4) {};
		\node [style=none] (3) at (0.5000002, -5) {};
	\end{pgfonlayer}
	\begin{pgfonlayer}{edgelayer}
		\draw (2.center) to (3.center);
		\draw (1.center) to (0.center);
	\end{pgfonlayer}
\end{tikzpicture}} \;\; \approx  \;\; \scalebox{0.6}{\input{doublemixmorphism2otherway_dubbelgevouwen_disconnect.tikz}}
\]
In other words, the identity $\circ$-separates, which is non-sense \cite{CoeckeBOOK}.
Therefore, the identity cannot be of form:
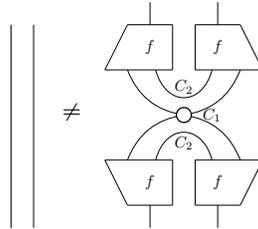
\[
\scalebox{0.6}{\begin{tikzpicture}
	\begin{pgfonlayer}{nodelayer}
		\node [style=none] (0) at (-0.5000002, -5) {};
		\node [style=none] (1) at (-0.5, 4) {};
		\node [style=none] (2) at (0.5000002, 4) {};
		\node [style=none] (3) at (0.5000002, -5) {};
	\end{pgfonlayer}
	\begin{pgfonlayer}{edgelayer}
		\draw (2.center) to (3.center);
		\draw (1.center) to (0.center);
	\end{pgfonlayer}
\end{tikzpicture}} \;\; \neq  \;\;\scalebox{0.6}{\input{doublemixmorphism2otherway.tikz}}
\]
With some wire-bending (using the Choi-Jamiolkowski isomorphism), we then see that the diagram composed of a cap followed by a cup cannot be of form:
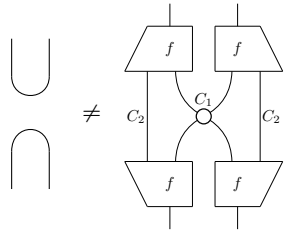
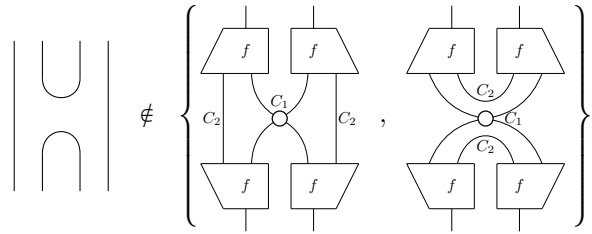
\[
\scalebox{1.0}{\input{capcomposedwithcup.tikz}}\;\; \neq \;\;\scalebox{0.6}{\input{doublemixmorphism2.tikz}}
\]
And so:
\[
\scalebox{1.0}{\input{notadoublemixture1.tikz}} \;\;\notin \;\;\left\{ \scalebox{0.6}{\input{doublemixmorphism2.tikz}}\;\;,\;\; \scalebox{0.6}{\input{doublemixmorphism2otherway.tikz}}  \right\}
\]
which completes the proof.
\end{proof}

\section{Double mixing $\subseteq$ double dilation}\label{section:subspace}

In the previous section, we gave an example of a vector resulting from double dilation that could not result from double mixing. The converse, however, does hold: every vector resulting from double mixing can be obtained from double dilation. In other words, double mixing yields a proper subspace of double dilation. The proof is a beautiful example of diagrammatic reasoning with spiders; for those who prefer bra-ket notation instead, the appendix contains a proof sketch in which we manipulate the bra-ket expression (\ref{result:doublemix}) to match that of (\ref{result:doublereduce}).

\begin{theorem}\label{theorem_subspace}
Every vector resulting from double mixing also results from double dilation.
\end{theorem}

\begin{proof}
Given a vector resulting from double mixing:
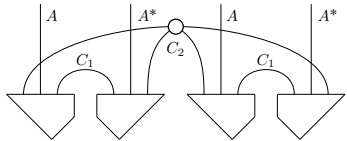
\[
\scalebox{0.6}{\input{doublemixstate.tikz}}
\]
We can use the spider fission (see section \ref{section:spiders}) to make four spiders:
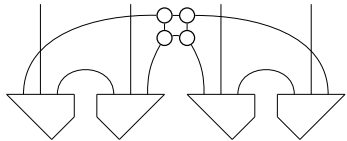
\[
\scalebox{0.6}{\input{doublemixstate_2.tikz}}
\]
Moving these spiders closer to the four boxes gives a familiar picture:
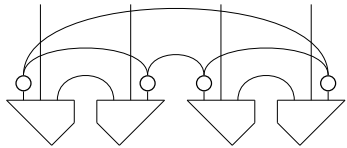
\[
\scalebox{0.6}{\input{doublemixstate_3.tikz}}
\]
In the final step, we absorb the spiders into the boxes, which yields the vector resulting from double mixing as a vector resulting from double dilation:
\[
\scalebox{0.6}{\input{doublemixstate_4.tikz}}
\]
\end{proof}

By Theorems \ref{theorem_strictsubspace} and \ref{theorem_subspace} we then obtain:

\begin{corollary}
Double mixing yields a proper subspace of double dilation.
\end{corollary}

\newpage
\section{Double mixing $\simeq$ disentangled states}\label{section:characterisation_twicemix}

\begin{defn}
Following \cite{CoeckeBOOK}, \em disentangled \em bipartite states are those states with diagrams of form:
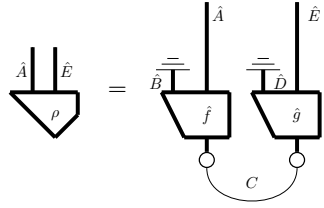
\[
\scalebox{0.6}{\input{rho_bipartite_state.tikz}} \;\; = \;\; \scalebox{0.6}{\input{general_disentangled_state.tikz}}
\]
\em Entangled \em states are those that are not disentangled.
\end{defn}
The intuition behind this idea is that disentangled states can share only classical information (the thin wire connecting the left and right halves of the diagram).

\begin{prop}
Vectors resulting from double mixing correspond to conjugate-symmetric
disentangled states.
\end{prop}

\begin{proof}
Consider a vector resulting from double mixing:
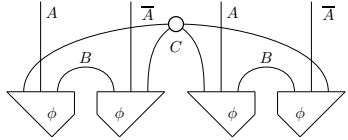
\[
\scalebox{0.6}{\input{doublemixstate_phi.tikz}}
\]
We rewrite it using the thick-line notation: $||\hat{\phi}\rangle\rangle = |\phi\rangle\overline{|\phi\rangle}$:
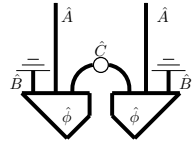
\[
\scalebox{0.6}{\input{doublemixstate_phi_thick.tikz}}
\]
Next, we move the spider downwards using the Choi-Jamiolkowski isomorphism, and use $\hat{f}$ for the operator resulting from the isomorphism applied to $||\hat{\phi}\rangle\rangle$. Then lastly, we use spider fission to arrive at:
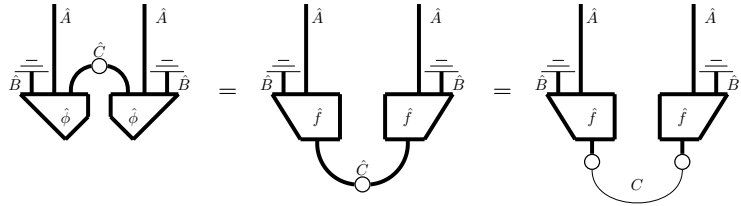
\[
\scalebox{0.6}{\input{doublemixstate_phi_thick.tikz}} \;\;=\;\; \scalebox{0.6}{\input{doublemixstate_phi_thick_down.tikz}} \;\;=\;\; \scalebox{0.6}{\input{doublemixstate_phi_thick_down_disentangled.tikz}}
\]
This is the general form of a disentangled state as described above, with only one restriction: the bipartite state has to be conjugate-symmetric.
\end{proof}

Contrast this with a vector resulting from double dilation:
Consider a vector resulting from double mixing and follow the same steps as above in rewriting:
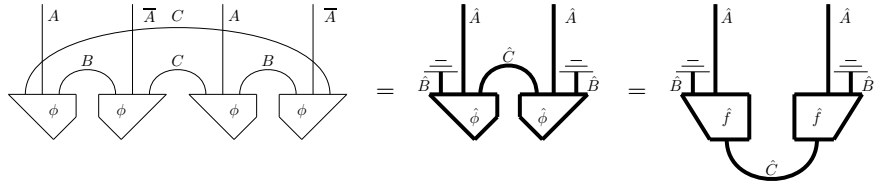
\[
\scalebox{0.6}{\input{dualdensitymatrix1_phi.tikz}} \;\; = \;\;
\scalebox{0.6}{\input{doubledilation_phi_thick.tikz}} \;\;=\;\; \scalebox{0.6} {\input{doubledilation_phi_thick_down.tikz}}
\]
This vector has the same symmetry as the one resulting from double mixing, but it is not necessarily disentangled. The symmetries introduced by mixing and dilation completely characterise the resulting vectors, which we show in the next section.

\section{The characterising symmetries of double dilation}\label{section:characterisation_twicedilation}

When we consider $|||\Psi\rangle\rangle\rangle$ from equation (\ref{result:doublemix}), there are two ways in which we can turn the vector $|||\Psi\rangle\rangle\rangle$ into an operator, both using the Choi-Jamiolkowski isomorphism:
\begin{align}
 |||\Psi\rangle\rangle\rangle = & \sum_{k=1}^{m}\sum_{i=1}^{n_k}\sum_{j=1}^{n_k} r_k p_{ki} p_{kj} |\phi_{ki}\rangle \overline{|\phi_{ki}\rangle}|\phi_{kj}\rangle \overline{|\phi_{kj}\rangle} \nonumber \\
  \text{operator}_1 = & \sum_{k=1}^{m}\sum_{i=1}^{n_k}\sum_{j=1}^{n_k} r_k p_{ki} p_{kj} |\phi_{ki}\rangle\overline{|\phi_{ki}\rangle} \langle\phi_{kj}| \overline{\langle\phi_{kj}|} \\
  \text{operator}_2 = & \sum_{k=1}^{m}\sum_{i=1}^{n_k}\sum_{j=1}^{n_k} r_k p_{ki} p_{kj} |\phi_{ki}\rangle|\phi_{kj}\rangle \langle\phi_{ki}| \langle\phi_{kj}|
\end{align}
Similarly for the doubly dilated vectors from equation (\ref{result:doublereduce}), for which we give the diagram expressions:
\begin{align}
  \scalebox{0.6}{\input{dualdensitymatrix1_phi.tikz}} \mapsto \;\; & \;\, \scalebox{0.6}{\input{densitymatrix1.tikz}} \;\; = \;\; \scalebox{0.6}{\input{dualdensitymatrix1_dm1.tikz}} \label{dm1} \\
   or\;\; \mapsto\;\; & \scalebox{0.6}{\input{densitymatrix2.tikz}} \;\; = \;\; \scalebox{0.6}{\input{dualdensitymatrix1_dm2.tikz}} \label{dm2}
\end{align}
In both cases, the resulting operator is positive semi-definite and self-adjoint. In other words, the operators are density operators. We call these \emph{the CJ-density operators of} the vector (from Choi-Jamiolkowski). The property of having two CJ-density operators completely characterises vectors resulting from double dilation.

\begin{theorem}\label{thm:characterisation2}
 Let $\phi$ be any normalised vector in any finite Hilbert space. Then $\phi$ has two CJ-density operators $CJ_1$ and $CJ_2$ iff $\phi$ is a result from double dilation.
\end{theorem}

\begin{proof}
We already argued that any result from double dilation has two CJ-density operators, so we only need to show the other direction: every normalised vector $\phi$ in a finite Hilbert space that has two CJ-density operators is the result from double dilation.

Suppose we are given a vector $\phi$
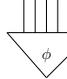
\[
\scalebox{0.6}{\input{4partitestatewithreducedtypes.tikz}}
\]
Notice that the information about $CJ_1$ and $CJ_2$ forces the type of $\phi$ to be of form $A\otimes \overline{A}\otimes A \otimes \overline{A}$.
We will now use the fact that an operator $\rho$ is positive semi-definite and self-adjoint iff it can be written as $\rho = f \circ f^{\dagger}$. Applying this  to $CJ_1$ and $CJ_2$, which are both positive semi-definite and self-adjoint, yields:
\[
CJ_1 = \;\; \scalebox{0.6}{\input{densitymatrix1_phi.tikz}} \;\; = \;\; \scalebox{0.6}{\input{ff_short.tikz}} \;\;\;\; \;\;\;\;\;\;\;\;\;\;\;\;\;\;\;\; CJ_2 = \;\; \scalebox{0.6}{\input{densitymatrix2_phi.tikz}} \;\; = \;\; \scalebox{0.6}{\input{gg_short.tikz}}
\]
where, when choosing a basis, $f = U_1 \sqrt{D_1} U_1^{-1}$ and $g = U_2 \sqrt{D_2} U_2^{-1}$, with $U_1$,$U_2$ unitaries and $D_1$, $D_2$ diagonal matrices (for the details on how tho get this form of $f$ and $g$, we refer to the proof in \cite[Theorem 7.2.1.]{Ashoush2015}).

We first show that both $CJ_1$ and $CJ_2$ are self-conjugate. As
\[
CJ_1 = \;\; \scalebox{0.6}{\input{densitymatrix1.tikz}}\;\; = \;\;\scalebox{0.6}{\input{dm1intermsofdm2.tikz}} \;\;=\;\; \scalebox{0.6}{\input{dm1intermsofdm22.tikz}}
\]

We have:
\[
CJ_1 = \;\;\scalebox{0.6}{\input{ff.tikz}}\;\; = \;\;\scalebox{0.6}{\input{ffintermsofgg.tikz}}
\]

Therefore, the conjugate of $CJ_1$ is:
\begin{align}
\overline{CJ_1} &=\;\; \scalebox{0.6}{\input{ffconjugate.tikz}}\;\; =\;\; \scalebox{0.6}{\input{ffintermsofggconjugate.tikz}} \;\;=\text{Pull!} \;\; \scalebox{0.6}{\input{ffintermsofggturning.tikz}}\;\; \text{Pull!} \\
& = \;\; \scalebox{0.6}{\input{ffintermsofgg.tikz}}\;\; =\;\; \scalebox{0.6}{\input{ff.tikz}} \;\;= CJ_1 \label{showingFselfconjugate}
\end{align}
Where we have used the fact that diagrams are equivalent as long as they are topologically the same \cite{CoeckeBOOK}. The diagram in between the `Pull!'s is not very meaningful on its own, but illustrates how to get to the next diagram.
We conclude that $CJ_1$ is self-conjugate. An analogous line of reasoning shows self-conjugateness for $CJ_2$.

We use this fact to mold $CJ_1$ into the right shape. Consider $f'$, defined as:
\[
f' = \;\; \scalebox{0.6}{\input{fprime.tikz}}
\]

$f'$ is positive semi-definite (inherited from $f$), and self-adjoint:

\begin{align}\label{fprimeisselfadj}
f'^{\dagger} = \;\;& \scalebox{0.6}{\input{fprimedagger.tikz}} \;\; = \;\; \scalebox{0.6}{\input{fprimedaggerconjugate.tikz}} \;\; = \;\; \scalebox{0.6}{\input{fprimedaggerconjugateturning.tikz}} \;\; = \;\; \scalebox{0.6}{\input{fprime.tikz}} \;\; = f'
\end{align}

Where in the second step we use that $f$ is self-conjugate, which follows from the fact that $CJ_1$ is so.
With $f'$ being positive semi-definite and self-adjoint, we can write $f'$ as:
\[
f' = \;\; \scalebox{0.6}{\input{hh_short.tikz}}
\]

And therefore we know:
\[
f = \;\; \scalebox{0.6}{\input{fintermsoffprime.tikz}} \;\; = \;\; \scalebox{0.6}{\input{fintermsofhh.tikz}}
\]

Substituting this expression into the equation for $CJ_1$:
\[
CJ_1 = \;\; \scalebox{0.6}{\input{ff.tikz}} \;\; = \;\; \scalebox{0.6}{\input{ffintermsofhhhh.tikz}}
\]

And therefore:
\[
\scalebox{0.6}{\input{4partitestatewithreducedtypes.tikz}} \;\; = \;\; \scalebox{0.6}{\input{hhhh.tikz}} \;\; = \;\; \scalebox{0.6}{\input{hhhhbendup.tikz}} \;\; = \;\; \scalebox{0.6}{\input{dualdensitymatrix1_notypes.tikz}}
\]

Which shows that $\phi$ is doubly dilated.
\end{proof}

So vectors resulting from double dilation are entirely characterised by their property of having two CJ-density operators.

\section{Multiple iterations}\label{section:generalisation}

We generalise mixing and dilation by iterating both constructions not just twice, but any finite number of times. The iterated version of dilation has already been thoroughly studied in Ashoush's Master thesis \cite{Ashoush}. Here, we still give a sketch of the results of each iteration, so we can contrast them with the results from iterated mixing. For mixing, we also just give the resulting diagrams, trusting that the reader can imagine how to generalise the construction given in section \ref{section:mixing} to more than two iterations.

The $0^{th}$ iteration of both constructions is just doing nothing, so we have normal vectors:
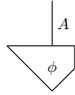
\[
\scalebox{0.6}{\input{phi_state.tikz}}
\]
Mixing then turns these vectors into ones of form:
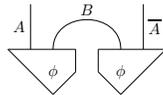
\[
\scalebox{0.6}{\input{reducedstate_nolabel.tikz}}
\]
The same goes for dilation:
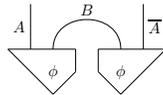
\[
\scalebox{0.6}{\input{reducedstate_nolabel.tikz}}
\]
The difference between the two shows in the second iteration. Double mixing is:
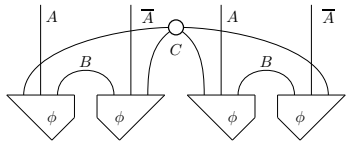
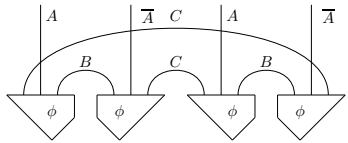
\[
\scalebox{0.6}{\input{doublemixstate_phi.tikz}}
\]
And double dilation:
\[
\scalebox{0.6}{\input{dualdensitymatrix1_phi.tikz}}
\]
We iterate both constructions a third time, first mixing:
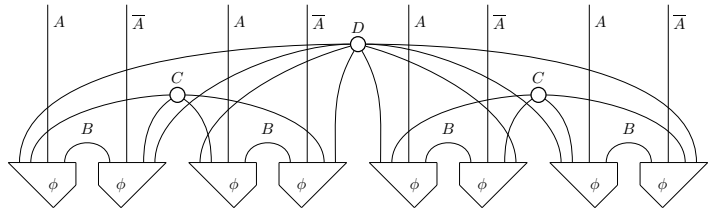
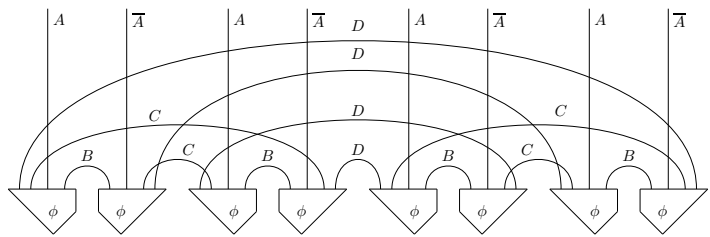
\[
\scalebox{0.6}{\input{3mixture_phi.tikz}}
\]
and then dilation:
\[
\scalebox{0.6}{\input{cpm3state_phi.tikz}}
\]

In general, the $n^{th}$ iteration of mixing takes the result from the previous iteration and tensors it with its conjugate. Then, all the resulting boxes are connected to a single new spider with $2^n$ legs.
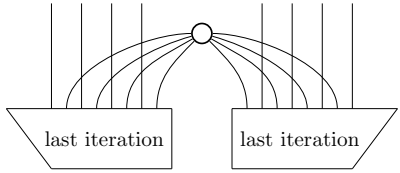
\[
\scalebox{0.8}{\input{general_mixing.tikz}}
\]
The $n^{th}$ iteration of dilation also takes the tensor product of the previous iteration and its conjugate, but then makes $2^{n-1}$ nested connections (a rainbow), connecting each box from the last iteration to its counterpart in the conjugate half.
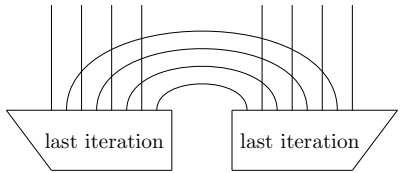
\[
\scalebox{0.8}{\input{general_cpm.tikz}}
\]

\subsection{Always a strict subspace}
In every iteration except for the first, mixing yields a proper subspace of dilation. This emphasises again that mixing and dilation are two non-equivalent constructions.
\begin{theorem}
 For all $n \geq 2$, $n$ iterations of mixing yields a proper subspace of the result from $n$ iterations of dilation.
\end{theorem}
\begin{proof}
The proofs of Theorems \ref{theorem_subspace} and \ref{theorem_strictsubspace} easily generalise to $n$ iterations instead of two.
\end{proof}

\subsection{More symmetry}\label{section:general_symmetry}

Vectors resulting from double dilation were characterised by having two CJ-density operators. This suggests that those resulting from $n$ iterations of dilation are precisely those that have $n$ CJ-density operators, capturing the extra symmetry introduced by each iteration of dilation.

\begin{theorem}\label{thm:characterisation_n}
Let $\phi$ be any vector in a finite Hilbert space. Then $\phi$ has $n$ CJ-density operators
$CJ_1, \ldots, CJ_n$ iff $\phi$ is the result of $n$ iterations of dilation.
\end{theorem}

The proof is by induction. Theorem \ref{thm:characterisation2} provides the base case, the rest of the induction is included in the appendix. Note that the fact that $\phi$ has $n$ CJ-density operators immediately implies that its type is of form $(A\otimes\overline{A})^{2^{n-1}}$, that is, it is a vector in Hilbert space $(A\otimes\overline{A})^{2^{n-1}}$.

Of course, as mixing always yields a subspace of dilation, vectors resulting from $n$ iterations of mixing also have the extra symmetry properties. They stay disentangled in the way discussed in section \ref{section:characterisation_twicemix}.

\section{Discussion and outlook}

Although in the physics community it us usually assumed that dilation and mixing are one and the same thing, this is clearly not the case.  The heart of our result can simply be depicted as:
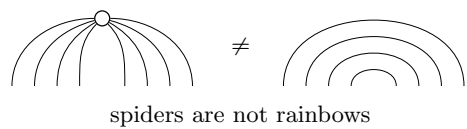
\begin{center}
  \begin{tabular}{ccc}
    \scalebox{0.6}{\input{general_spider_8legs.tikz}} & \raisebox{5mm}{$\neq$} & \scalebox{0.6} {\input{general_rainbow.tikz}} \\
      \multicolumn{3}{c}{\small spiders are not rainbows} \\
  \end{tabular}
\end{center}
That is, a convex sum over pure operators (mixing, yielding a spider diagram) is not the same as a partially traced out composite system (dilation, yielding a rainbow diagram), even though both constructions happen to coincide in the case of density operators (i.e.~the result of applying them only once), since:
\begin{center}
  \begin{tabular}{ccc}
    \scalebox{0.6}{\begin{tikzpicture}
	\begin{pgfonlayer}{nodelayer}
		\node [style=none] (0) at (-2, 0) {};
		\node [style=none] (1) at (2, 0) {};
		\node [style=wn] (2) at (0, 1.5) {};
	\end{pgfonlayer}
	\begin{pgfonlayer}{edgelayer}
		\draw [in=180, out=90, looseness=1.00] (0.center) to (2);
		\draw [in=90, out=0, looseness=1.00] (2) to (1.center);
	\end{pgfonlayer}
\end{tikzpicture}} & \raisebox{2.5mm}{$=$} & \scalebox{0.6}{\begin{tikzpicture}
	\begin{pgfonlayer}{nodelayer}
		\node [style=none] (0) at (-2, 0) {};
		\node [style=none] (1) at (2, 0) {};
	\end{pgfonlayer}
	\begin{pgfonlayer}{edgelayer}
		\draw [bend left=90, looseness=1.25] (0.center) to (1.center);
	\end{pgfonlayer}
\end{tikzpicture}} \\
  \end{tabular}
\end{center}

In physics, this result may impact axiomatic understanding of density matrices, and may also contribute to either crafting interesting toy theories, or adjoining extra variables to theories.

In NLP, it is worth considering which of the two, double mixing and double dilation, could serve as a model for both ambiguity and lexical entailment. Notice that in dictionaries, disambiguation of words is always first by hypernym, then by entailment. If this order is something that the model should reflect, then double mixing is a good choice: the asymmetry is reflected by the spiders appearing in the mixtures, causing a clear distinction between the first and second iterations of mixing. If however, this order of disambiguation in dictionaries is considered artificial, then the more general double dilation might be the preferred option.

The second result presented in this paper is the characterisation of both constructions. Dilation yields vectors that have $n$ CJ-density operators, which nicely exposes the symmetries introduced by the construction. Mixtures on the other hand, while having the same symmetries as dilated vectors, are special cases of disentangled states. This actually comes as no surprise: mixtures are almost by definition impure things. For future research, it would be ideal to find a characterisation for vectors resulting from double dilation that are \em not \em the result of double mixing.

As we mentioned in section \ref{section:diagram_notation}, the results in this paper apply in a more general setting than finite Hilbert spaces. To be precise, they hold in any \em spider category \em (dagger compact closed category with a Frobenious structure). One such category is the category of sets and relations (Rel). Oscar Cunningham and Dan Marsden have looked into the application of iterated dilation to the states in Rel\cite{Marsden15,Cunningham-privite}. In ongoing research, we are now applying iterated mixing to Rel as well. Hopefully, this will give us some hints about vectors resulting from double dilation but not from double mixing.

\bibliography{library}{}
\bibliographystyle{plain}

\appendix

\section{Proofs and proof sketches}

\subsection*{From section \ref{section:subspace}}

We sketch the proof of Theorem \ref{theorem_subspace} in terms of bras and kets.
Consider a vector resulting from double mixing:
\[
|||\Psi\rangle\rangle\rangle = \sum_{k=1}^{m}\sum_{i=1}^{n_k}\sum_{j=1}^{n_k} r_k p_{ki} p_{kj} |\phi_{ki}\rangle \overline{|\phi_{ki}\rangle}|\phi_{kj}\rangle \overline{|\phi_{kj}\rangle}
\]
In order to compare $|||\Psi\rangle\rangle\rangle$ to vectors that resulting from double dilation, we choose two Hilbert spaces $B,C$ such that $\dim(C) = m$ and $\dim(B)= N = \max_k(n_k)$, and choose orthonormal bases $\{e_i^B\}$ and $\{e_k^C\}$ for $B$ and $C$ respectively. We then define $|\phi'_{ki}\rangle$ for each $k \leq m$ and $i \leq N$ as follows:
\[
|\phi'_{ki}\rangle = \begin{cases}
                        \sqrt[4]{r_k}\sqrt{p_i} |e^B_i\rangle |e^C_k\rangle |\phi_{ki}\rangle, & \mbox{if $i \leq n_k$ ($|\phi_{ki}\rangle$ exists)} \\
                        0 |e^B_i\rangle|e^C_k\rangle, & \mbox{otherwise}.
                      \end{cases}
\]
so that $\langle e^C_k|\langle e^B_i|\phi'_{ki}\rangle = \sqrt[4]{r_k}\sqrt{p_i}|\phi_{ki}\rangle$ (or 0 if $|\phi_{ki}\rangle$ does not exist), and hence:
\[
|||\Psi\rangle\rangle\rangle = \sum_{k=1}^{m}\sum_{i=1}^{N}\sum_{j=1}^{N} \langle e^C_k|\langle e^B_i|\phi'_{ki}\rangle\overline{\langle e^C_k|\langle e^B_i|\phi'_{ki}\rangle}\langle e^C_k|\langle e^B_j|\phi'_{kj}\rangle \overline{\langle e^C_k|\langle e^B_j|\phi'_{kj}\rangle}
\]
Next, we define:
\[
|\phi'_{ABC}\rangle = \sum_{k=1}^{m} \sum_{i=1}^{N} |\phi'_{ki}\rangle
\]
Because the bases $\{e_k^C\}$ and $\{e_i^B\}$ are orthonormal, we still have $\langle e^C_k|\langle e^B_i|\phi'_{ABC}\rangle = |\phi'_{ki}\rangle$ and hence:
\[
|||\Psi\rangle\rangle\rangle = \sum_{k=1}^{m}\sum_{i=1}^{N}\sum_{j=1}^{N} \langle e^C_k|\langle e^B_i|\phi'_{ABC}\rangle\overline{\langle e^C_k|\langle e^B_i|\phi'_{ABC}\rangle}\langle e^C_k|\langle e^B_j|\phi'_{ABC}\rangle \overline{\langle e^C_k|\langle e^B_j|\phi'_{ABC}\rangle}
\]

We apply a final tweak using the Dirac delta to arrive at an expression matching that of a doubly dilated vector:
\begin{align}\label{result:doublemix_subset_doublereduce}
 |||\Psi\rangle\rangle\rangle = & \sum_{k=1}^{m}\sum_{l=1}^{m}\sum_{i=1}^{N}\sum_{j=1}^{N} \delta_{kl} \langle e^C_k|\langle e^B_i|\phi'_{ABC}\rangle\overline{\langle e^C_l|\langle e^B_i|\phi'_{ABC}\rangle}\langle e^C_l|\langle e^B_j|\phi'_{ABC}\rangle \overline{\langle e^C_k|\langle e^B_j|\phi'_{ABC}\rangle}
\end{align}

Therefore, every vector resulting from double mixing can be written as one resulting from double dilation, with a restriction given by the Dirac delta. We conclude that doubly mixed vectors from a strict subspace of doubly dilated vectors.

\subsection*{From section \ref{section:general_symmetry}}

We give the proof of Theorem \ref{thm:characterisation_n}:\\
Let $\phi$ be any vector in a finite Hilbert space. Then $\phi$ has $n$ CJ-density operators
$CJ_1, \ldots, CJ_n$ iff $\phi$ is the result of $n$ iterations of dilation.

\begin{proof}
As in Theorem \ref{thm:characterisation2}, the symmetries of vectors resulting from $n$ iterations of dilation take care of the $\Leftarrow$ direction, so we only need to show the $\Rightarrow$ direction. The proof depends heavily on the Choi-Jamiolkowski isomorphism. As the proof is diagrammatic, this is referred to as `wire bending'.

Let $\phi$ be any vector that has CJ-density operators $CJ_1 \ldots CJ_n$. This immediately implies that $\phi$ must be of type $(A \otimes \overline{A})^{n-1}$. We use the following systematic enumeration of all the CJ-density operators:
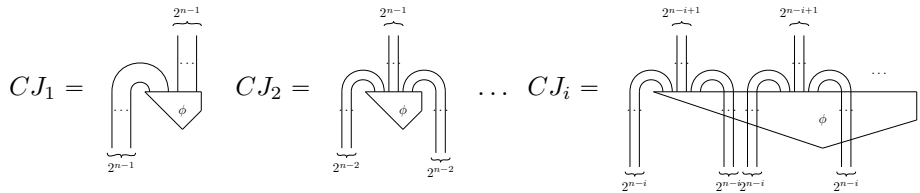
\[
CJ_1 = \;\; \scalebox{0.5}{\input{cpmn_densitymatrix1.tikz}} \;\;\;\; CJ_2 = \;\;  \scalebox{0.5}{\input{cpmn_densitymatrix2.tikz}} \;\; \ldots \;\; CJ_i = \;\; \scalebox{0.5}{\input{cpmn_densitymatrix_generalcase.tikz}}
\]
The general pattern is: $CJ_i$ is obtained from $\phi$ by executing the following procedure: 1) Bend down the first bundle of $2^{n-i}$ outputs. 2) Leave the next $2^{n+1-i}$ outputs up. 3) Bend down the next bundle of $2^{n-i}$ outputs and make sure that the whole bundle ends up to the right of the already existing inputs. 4) Repeat step 3. 5) Repeat step 2 - 4 until you have considered all outputs of $\phi$.
\vskip 2ex
Shaping $\phi$ into a vector resulting from $n$ iterations of dilation consists of two main parts. First we break $\phi$ into $2^n$ pieces and then we put the pieces back together so that they form the right symmetries. We describe both procedures step by step:
\begin{enumerate}
  \item Breaking $\phi$ apart:
     \begin{enumerate}
       \item Consider $CJ_1$, the density operator formed by bending down the first $2^{n-1}$ outputs of $\phi$. As this is a density operator, write it as $f^{\dagger}_n \circ f_n$:
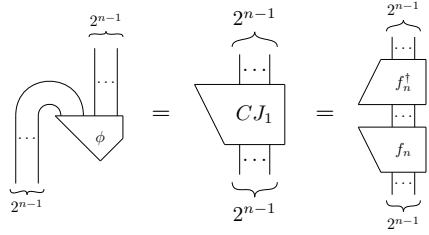
           \[
           \scalebox{0.6}{\input{cpmn_densitymatrix1.tikz}} \;\; = \;\; \scalebox{0.8}{\input{cpmn_phi1.tikz}}\;\; = \;\;\scalebox{0.6}{\input{cpmn_fnfndagger.tikz}}
           \]
      \item Define $f'_n$ as follows (see diagram). Start with $f$ and divide the in- and outputs of $f$ in groups of four. For every such group of four, bend the outer two wires of the inputs up and the inner two wires of the outputs down, as illustrated in the first diagram below. Notice that this is the same bending of wires that would be needed to form $CJ_{n-1}$ into $CJ_n$. This will yield a map that is self-conjugate. Then, bend the first half of the inputs of this resulting map upwards, and the second half of the outputs of this map downwards, to create a self-adjoint map (the second step in the diagram below). This is $f'_n$.

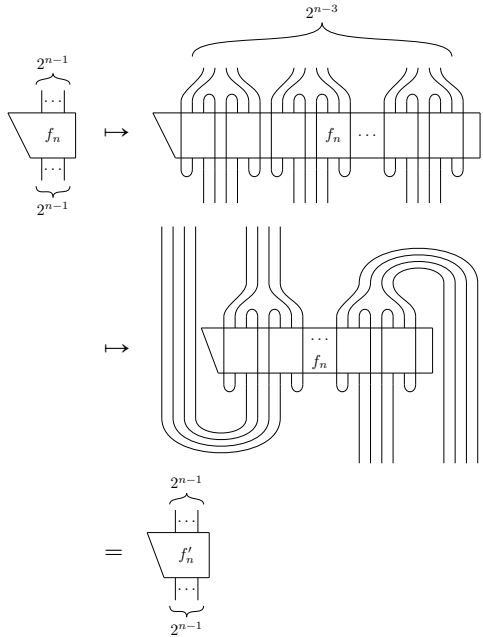
          \begin{align*}
            \scalebox{0.6}{\input{cpmn_fn.tikz}}\;\; & \mapsto \;\; \scalebox{0.6}{\input{cpmn_fnprime_step1.tikz}}\\
             & \mapsto \;\;\ \scalebox{0.6}{\input{cpmn_fnprime_step2.tikz}} \\
             & = \;\; \scalebox{0.6}{\input{cpmn_fnprime.tikz}}
          \end{align*}

          The fact that $f'_n$ is self-adjoint follows from the same line of reasoning used in equation \ref{fprimeisselfadj} in Theorem \ref{thm:characterisation2}.

      \item Since $f'_n$ is self-adjoint, write it as $f^{\dagger}_{n-1} \circ f_{n-1}$.

          \[
          \scalebox{0.6}{\input{cpmn_fnprime.tikz}} \;\; = \;\; \scalebox{0.6}{\input{cpmn_fnmin1fnmin1dagger.tikz}}
          \]

      \item Define $f'_{n-1}$ similarly to $f'_n$, now using the same wire-bending that would be needed to turn $CJ_{n-2}$ into $CJ_{n-1}$.
      \item Since $f'_{n-1}$ is self-adjoint, write it as $f^{\dagger}_{n-2} \circ f_{n-2}$.
      \item Repeat the previous two steps to create $f_{n-3}, f_{n-4}, \ldots$ until you have $f_2$. Every $f'_{n-i}$ is obtained from $f_{n-i}$ using the wire-bending that would turn $CJ_{n-(i+1)}$ into $CJ_{n-i}$.
      \item Define $f'_2$ by only bending the first half of the inputs of $f_2$ upwards, and the second half of the outputs of $f_2$ downwards. This is again a self-adjoint map.
      \item Since $f'_2$ is self-adjoint, write it as $f^{\dagger}_1 \circ f_1$.
    \end{enumerate}
  \item We can now build $\phi$ from $2^n$ copies of $f_1$. From the procedure above, we have the following equations:
      \begin{align*}
        CJ_1 = & f^{\dagger}_n \circ f_n \\
        f'_n = & \text{ wire-bent version of } f_n \\
        f'_n = & f^{\dagger}_{n-1} \circ f_{n-1} \\
        f'_{n-1} = & \text{ wire-bent version of } f_{n-1} \\
         & \vdots \\
        f'_2 = & \text{ wire-bent version of } f_2 \\
        f'_2 = & f^{\dagger}_1 \circ f_1
      \end{align*}
      By just reversing the wire-bending parts (i.e., expressing $f_2$ as a wire-bent version of $f'_2$ etc), we get a chain of equations that we can substitute into each other, eventually expressing $CJ_1$ entirely in terms of $f_1$:
      \begin{align*}
      f'_2 = & f^{\dagger}_1 \circ f_1 \\
      f_2 = & \text{ wire-bent version of } f'_2 \\
      f_3' = & f^{\dagger}_2 \circ f_2 \\
      & \vdots \\
      f_{n-1} = & \text{ wire-bent version of } f'_{n-1} \\
      f'_n = & f^{\dagger}_{n-1} \circ f_{n-1} \\
      f_n = & \text{ wire-bent version of } f'_n \\
      CJ_1 = & f^{\dagger}_n \circ f_n
      \end{align*}

      Then, all that is left to do is form $CJ_1$ back into $\phi$, which is done by bending all the wires up and carefully arranging the $2^n$ copies of $f_1$. The result is that we have written $\phi$ as a vector resulting from $n$ iterations of dilation.
\end{enumerate}

We prove correctness of this procedure by induction. In the case of double dilation, the procedure coincides with the proof given in Theorem \ref{thm:characterisation2}, which provides the basis for the induction. We may hence assume that the procedure correctly shows that for a vector $\phi$, positive semi-definiteness and self-adjointness of $CJ_1 \ldots CJ_{n-1}$ implies that $\phi$ is a vector resulting from $n-1$ iterations of dilation (we refer to this as `case $n-1$'). We need to show that if in addition $CJ_n$ is positive semi-definite and self-adjoint, then the procedure shows that $\phi$ is the result from $n$ iterations of dilation (`case $n$').

By the induction hypothesis, we may assume that $f_{n-1}$ as defined in the procedure for case $n-1$ consists of $2^{n-2}$ copies\footnote{recall that for the final result, we need $f^{\dagger}_{n-1} \circ f_{n-1}$ to consist of $2^{n-1}$ copies of $f_1$, hence $f_{n-1}$ itself consists of $2^{n-2}$ copies.} of $f_1$, and each copy has an output wire that is either an in- or output of $f_{n-1}$.

Notice that when we execute the procedure for case $n$, the definitions of $f_{n-1}, \ldots, f_1$ are identical to the definitions of $f_{n-1}, \ldots, f_1$ in the procedure of case $n-1$, as are the expressions of $f_2, \ldots, f_{n-1}$ in terms of $f_1$. The difference between the two procedures is the type that we assume for the in- and outputs of these morphisms. In case $n-1$, we may assume that $\phi$ has type $(A \otimes \overline{A})^{n-1}$, whereas in case $n$, it has type $(A \otimes \overline{A})^{n}$. The effect of this difference in types in that every wire drawn in a diagram from the procedure of case $n-1$ can be replaced by a \emph{pair} of wires to get the same diagram in the procedure of case $n$. This is illustrated below for $f_{n-1}$.

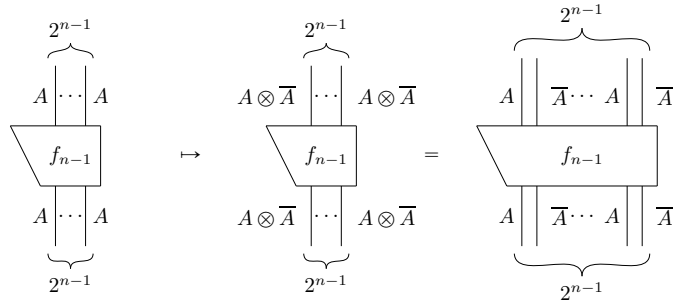
\[
\scalebox{0.8}{\input{cpmn_case_nmin1_to_case_n.tikz}}
\]

So, in case $n$, $f_{n-1}$ consists of $2^{n-2}$ copies of $f_1$, and each copy has a pair of output wires that are either both inputs or both outputs of $f_{n-1}$. $f'_n$, which equals $f^{\dagger}_{n-1} \circ f_{n-1}$, turns the $2^{n-2}$ copies of $f_1$ into $2^{n-1}$ copies, keeping the pais of wires together.

Now comes the magic.
The transformation from $f'_n$ to $f_n$ splits the pairs of wires, causing the first half of the pair to become an output in $f_n$, and the second half of the pair to become an input of $f_n$. As $CJ_1 = f^{\dagger}_n \circ f_n$, the outputs of $f_n$ are connected to the inputs of $f^{\dagger}_n$ and form the rainbow $C_n$ in the final result. Because of the symmetry between $f_n$ and its adjoint $f^{\dagger}_n$, the individual rainbow arcs connect corresponding copies of $f_1$ (see Figure \ref{illustration_of_connections} below). The inputs of $f_n$ and the outputs of $f^{\dagger}_n$ become the output wires of the final result. The result hence has the structure of a vector resulting from $n$ iterations of dilation, with every output of $\phi$ coming from a different copy of $f_1$.

\begin{figure}
  \centering
  \[
   \scalebox{0.7}{\input{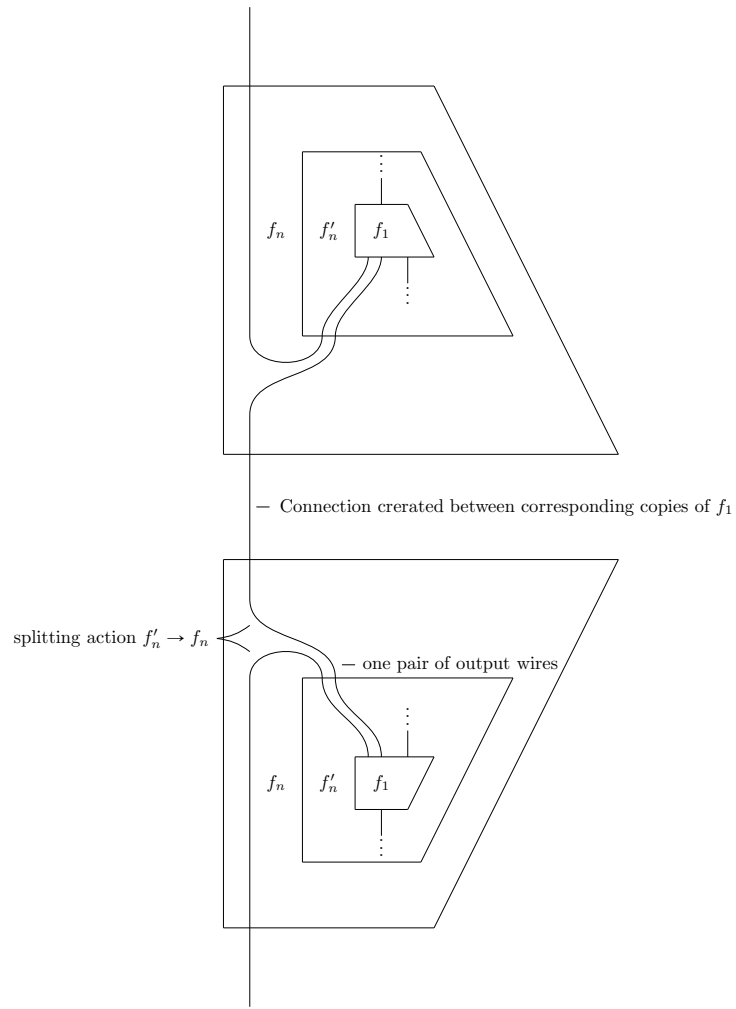}}
  \]
  \caption{Illustration of the connections made between the copies of $f_1$ in $f_n$ and $f^{\dagger}_n$.}\label{illustration_of_connections}
\end{figure}
\end{proof}

\end{document}